\documentclass[oneside]{amsart}

\makeatletter
\g@addto@macro{\endabstract}{\@setabstract}
\newcommand{\authorfootnotes}{\renewcommand\thefootnote{\@fnsymbol\c@footnote}}%
\makeatother

\usepackage{url}

\usepackage[utf8]{inputenc}
\usepackage[T1]{fontenc}
\usepackage{amsthm}
\usepackage{mathtools}
\usepackage{amsfonts}
\usepackage{amsmath}
\usepackage{amssymb}
\usepackage{xfrac}
\usepackage{enumitem}
\usepackage{graphicx}
\usepackage{mathrsfs}
\usepackage{subcaption}
\usepackage[margin=3cm]{geometry}
\usepackage{tikz}
\usetikzlibrary{arrows.meta}

\usepackage{etoolbox}

\usepackage[
unicode=true,
bookmarks=true,
bookmarksopen=false,
breaklinks=false,
pdfborder={0 0 0},
colorlinks=true
]{hyperref}

\usepackage{xcolor}
\definecolor{cblue}{rgb}{0.16, 0.32, 0.75}
\definecolor{cred}{rgb}{0.7, 0.11, 0.11}
\hypersetup{%
	,linkcolor=cred
	,citecolor=cblue
	,urlcolor=cred
}

\usepackage[
backend=biber,
style=alphabetic,
labelnumber,
defernumbers=true,
giveninits=true,
maxbibnames=299,
natbib=true,
doi=false,
isbn=false,
url=false,
sorting=nyt,
eprint=false,
date=year
]
{biblatex}
\renewbibmacro{in:}{}
\DeclareFieldFormat[misc]{title}{``#1''\isdot}
\DeclareFieldFormat{journaltitle}{#1\isdot}
\DeclareFieldFormat[article,periodical]{volume}{\mkbibbold{#1}}
\DeclareFieldFormat{pages}{#1}

\AtEveryBibitem{%
	\clearfield{number}}

\AtEveryBibitem{ \clearlist{language} }

\bibliography{Stability}

\renewcommand{\d}{\mathrm{d}}
\newcommand{\D}{\mathcal{D}}
\renewcommand{\H}{\mathcal{H}}
\newcommand{\R}{\mathbb{R}}

\newcommand{\norm}[1]{\left\Vert #1 \right\Vert}
\newcommand{\argdot}{{\hspace{0.18em}\cdot\hspace{0.18em}}}
\newcommand{\scalar}[2]{\langle #1, #2\rangle}
\renewcommand{\Re}{\mathrm{Re}\,}
\renewcommand{\Im}{\mathrm{Im}\,}
\newcommand{\nfrac}[2]{{}^{#1}{\mskip -5mu/\mskip -3mu}_{#2}}

\newtheorem{theorem}{Theorem}[section]
\newtheorem{corollary}[theorem]{Corollary}
\newtheorem{lemma}[theorem]{Lemma}
\newtheorem{proposition}[theorem]{Proposition}
\theoremstyle{definition}
\newtheorem{definition}[theorem]{Definition}
\newtheorem{assumption}[theorem]{Assumption}
\theoremstyle{remark}
\newtheorem{remark}[theorem]{Remark}

\numberwithin{equation}{section}

\newcommand{\triang}{\hfill$\triangle$}
\AtEndEnvironment{example}{\triang}

\newcommand{\rev}[1]{{\color{blue} #1}}

\title[Stability of non-autonomous Schrödinger equations]{}

\subjclass[2020]{{35Q41, 35J10, 37K45, 81Q93}}

\begin{document}
	
	\begin{center}
		\LARGE
		On a sharper bound on the stability of non-autonomous Schrödinger equations and applications to quantum control
		\par \bigskip
		
		\normalsize
		\authorfootnotes
		Aitor Balmaseda\footnote{abalmase@math.uc3m.es}\textsuperscript{,1,2},  
		Davide Lonigro\footnote{davide.lonigro@fau.de}\textsuperscript{,3,4,5}, and
		Juan Manuel Pérez-Pardo\footnote{jmppardo@math.uc3m.es}\textsuperscript{,1} \par \bigskip
		
		\textsuperscript{1}\footnotesize Departamento de Matemáticas, Universidad Carlos III de Madrid, Avda. de la Universidad 30, 28911  Madrid, Spain \par
		\textsuperscript{2}\footnotesize Departamento de Análisis Matemático y Matemática Aplicada, Facultad de Ciencias Matemáticas, Universidad
		Complutense de Madrid, Madrid 28040, Spain \par
		\textsuperscript{3}\footnotesize Department Physik, Friedrich-Alexander-Universität Erlangen-Nürnberg, Staudtstraße 7/B3, 91058 Erlangen, Germany \par
		\textsuperscript{4}\footnotesize Dipartimento di Matematica, Università di Bari Aldo Moro, I-70125 Bari, Italy \par
		\textsuperscript{5}\footnotesize Istituto Nazionale di Fisica Nucleare, Sezione di Bari, I-70126 Bari, Italy \par
		\par \bigskip
		
		\today
	\end{center}
	
	\begin{abstract}
		
		We study the stability of the Schrödinger equation generated by time-dependent Hamiltonians with constant form domain. That is, we bound the difference between solutions of the Schrödinger equation by the difference of their Hamiltonians. The stability theorem obtained in this article provides a sharper bound than those previously obtained in the literature. This makes it a potentially useful tool for time-dependent problems in Quantum Physics, in particular for Quantum Control. We apply this result to prove two theorems about global approximate controllability of infinite-dimensional quantum systems. These results improve and generalise existing results on infinite-dimensional quantum control.
	
	\end{abstract}
	
	\maketitle
	\vspace{-1cm}
	
	
	\section{Introduction}
	
	In this article we study the problem of existence of solutions of the non-autonomous Schrödinger equation, with particular emphasis on the stability of the solutions under perturbations of the generating Hamiltonians. We have obtained a sharper bound that has enabled us to prove several results on the controllability of quantum systems whose state space is determined by an infinite-dimensional Hilbert space. 
	
	{
	For doing so, we have addressed the problem
	in the abstract setting,
	relying on results on the properties of Hamiltonians with constant form domain~\cite{Simon1971} and related results on the existence of solutions of the non-autonomous Schrödinger equation~\cite{Kisynski1964, balmaseda2022schrodinger}.
	}
	This is a research field that developed during the 50's and 60's, see e.g.~\cite{Yosida1948, Yosida1963}, and T.\ Kato's celebrated result~\cite{Kato1973}. Since then, many articles refining the results in more concrete cases have appeared, e.g.~\cite{Yajima1987, burgarth2022one} which consider Hamiltonians of Schrödinger type.
	We have provided bounds, see Lemma~\ref{lemma:propagators_bound_simon}, that improve previous results in the literature~\cite{Simon1971} and which allowed us to obtain a novel abstract stability result, Theorem~\ref{thm:stability_bound}. The bound obtained is sharper than those obtained 
	in the literature, cf.~\cite[Theorem III]{Kato1973} and~\cite[Theorem 9 and Corollary 10]{Sloan1981}. The main improvement relies on the fact that our stability result allows for discontinuous generators.
	We also apply this to the particular case of form-linear Hamiltonians, cf.\ Definition~\ref{def:formlinear_hamiltonian}, to obtain Theorem~\ref{thm:stability_formlinear}.
	The latter is a stability result on a collection of non-autonomous Schrödinger equations that appears often in the applications.
		
	Furthermore, we apply the stability results obtained to prove several results about approximate controllability in the context of Quantum Control for infinite-dimensional systems, i.e.\ defined on a complex separable Hilbert space with infinite dimension. This research field has attracted interest over the past decades due to {its connection to quantum computation and information, and its direct applications to the development of quantum technologies}~\cite{koch2022quantum}. Remarkably, exact controllability (see Definition~\ref{def:exact_control}), in contrast with what happens in the finite-dimensional setting~\cite{DAlessandro2007}, is not a convenient notion of controllability in infinite-dimensional quantum systems as the different characterisations of the reachable sets establish~\cite{BallMarsdenSlemrod1982, Turinici2000, boussaid2020regular}. 
	Less restrictive notions, like approximate controllability (see Definition~\ref{def:approx_control}) or exact controllability in dense subsets, have been employed to prove results in interesting particular cases; for instance, local exact controllability~\cite{Beauchard2005a, BeauchardCoron2006a, beauchard2010local} or global approximate controllability~\cite{nersesyan2010global}, as well as other controllability notions specific to the quantum case~\cite{IbortPerezPardo2009}.
	{Abstract results that can be applied on generic situations are still scarce: among these,
	results of remarkable interest obtained via Lie--Galerkin control techniques, consisting in sufficient conditions for the approximate controllability of bilinear control systems via piecewise constant control functions~\cite{ChambrionMasonSigalottiEtAl2009, MasonSigalotti2010, BoscainCaponigroChambrionEtAl2012,caponigro2018exact,boussaid2023controllability},}
	have been obtained and applied successfully, for instance, to the control of molecules~\cite{BoussaidCaponigroChambrion2019, BoscainPozzoliSigalotti2021}. To complete this overview, it is also worth mentioning the successful approach based on adiabatic techniques to achieve approximate controllability~\cite{augier2022effective, duca2023control, IbortMarmoPerezPardo2014a, robin2022ensemble}.
	
		The control results obtained in this article, Theorem~\ref{thm:smooth_controls} and Theorem~\ref{thm:Lie-Galerkin control}, are susceptible to be applied on a wide variety of situations: these will be discussed at the end of Section~\ref{sec:applications} and Section~\ref{sec:remarks}, and are in direct relation with the Lie--Galerkin techniques that pursue to extend the well developed theory of Quantum Control on finite-dimensional systems to infinite-dimensional ones. Indeed, for finite-dimensional quantum systems, geometric control theory~\cite{AgrachevSachkov2004, Jurdjevic1996}, provides a rich framework that has been used fruitfully, as the development of the quantum technologies shows: however, these results rely on piecewise constant control functions, and the current stability results cannot be directly applied to them. For instance, in the most general case, cf.~\cite[Theorem III]{Kato1973}, they require the existence of an upper integrable function $\dot{u}(t)$ such that ${u}(t)=\int_0^t\dot{u}(t')\mathrm{d}t'$; or, in other results like~\cite{Sloan1981}, continuously differentiable controls. The control results presented in this article lift the restrictions on the interaction terms obtained in~\cite{nersesyan2010global}, i.e., $L^2$-bounded interactions, and also in~\cite{boussaid2013weakly}, for the case of $k=1$ of $k$-weakly coupled systems, cf.\ the discussion in Section~\ref{sec:remarks}. 
		The improved bounds for the stability obtained  in this article play a crucial role in the proofs of such results. We are able to prove approximate controllability results provided that the finite dimensional approximations are exactly controllable with controls that are uniformly bounded in $L^1$ norm, which is in accordance with recent results in the literature, cf.\ \cite{caponigro2018exact, boussaid2023controllability}.
	
	The article is organised as follows. In Section~\ref{sec:preliminaries} we introduce the mathematical framework, recall some important results, and fix the notation. In  Section~\ref{sec:form_dynamics} we present some results about the existence of solutions of the Schrödinger equation and obtain the main stability result of the paper, Theorem~\ref{thm:stability_bound} as well as its application to the particular case of form-linear Hamiltonians, Theorem~\ref{thm:stability_formlinear}. In Section~\ref{sec:applications} we present the problem of quantum control for infinite-dimensional systems, give the main notions about controllability and prove two controllability results on infinite-dimensional quantum systems, Theorem~\ref{thm:smooth_controls} and Theorem~\ref{thm:Lie-Galerkin control}. We also derive some consequences of them.
	We conclude with Section~\ref{sec:remarks}, in which we summarise the main results of this article and put them in relation with previous results in the literature.	
	
	
	\section{Mathematical Preliminaries}\label{sec:preliminaries}	
	
	The non-autonomous Schrödinger equation
	is a linear evolution equation on a Hilbert space $\mathcal{H}$ defined by a family of self-adjoint operators $\{H(t)\}_{t \in\R}$, densely defined on subsets of $\mathcal{H}$, which is called time-dependent Hamiltonian. To simplify the notation, we will denote this family as $H(t)$.	
	
	In the most general case, the time-dependent Hamiltonian is a family of unbounded self-adjoint operators whose domains depend on $t$; these facts make the problem of the existence of solutions for the Schrödinger equation highly non-trivial. There are general sufficient conditions for the existence of solutions of these equations, see for instance~\cite{Kato1973,Kisynski1964,Simon1971}. An important notion to address this problem is the notion of sesquilinear form associated with an operator. We recall next an important result, cf.\ \cite[Sec. VI.2]{Kato1995}.
	
	\begin{definition}\label{def:graph-norm}
		Let $h$ be a Hermitian sesquilinear form with dense domain $\D(h)$. If $h$ is semibounded from below with 
		{lower bound} $-m$, $m\geq 0$, we define the \emph{graph norm} of $h$ by
		$$\norm{\Phi}_h:=\sqrt{(1+m)\norm{\Phi}^2+h(\Phi,\Phi)},\quad \Phi\in\D(h).$$
		{The form $h$ is \textit{closed} if $\mathcal{D}(h)$, endowed with the graph norm, is a Hilbert space.}
	\end{definition}
	
	\begin{theorem}[Representation Theorem]\label{thm:repKato}
		Let $h$ be an Hermitian, closed, semibounded sesquilinear form densely defined on $\mathcal{D}(h) \subset \mathcal{H}$.
		Then there exists a unique, self-adjoint, semibounded operator $T$ with domain {$\mathcal{D}(T)$} and the same lower bound such that:
		\begin{enumerate}[label=\textit{(\roman*)},nosep, leftmargin=*]
			\item $\Phi \in \mathcal{D}(T)$ if and only if $\Phi \in \mathcal{D}(h)$ and there exists
			$\chi \in \mathcal{H}$ such that
			\begin{equation*}
				h(\Psi, \Phi) = \langle \Psi, \chi \rangle, \qquad \forall \Psi \in \mathcal{D}(h).
			\end{equation*}
			\item $h(\Psi, \Phi) = \langle \Psi, T\Phi \rangle$ for any $\Psi \in \mathcal{D}(h)$,
			$\Phi \in \mathcal{D}(T)$.
			\item $\mathcal{D}(T)$ is a core for $h$, that is,
			$\overline{\mathcal{D}(T)}^{\norm{\argdot}_h} = \mathcal{D}(h)$.
		\end{enumerate}
	\end{theorem}
	
	{Throughout this paper, we will be interested in evolution problems generated by the following kind of time-dependent Hamiltonians:}
	\begin{definition} \label{def:hamiltonian_const_form}
		Let $I \subset \mathbb{R}$ be an interval, let $\mathcal{H}^+$ be a dense subspace of $\mathcal{H}$, and $H(t)$, ${t \in I}$, a family of self-adjoint operators on $\mathcal{H}$ such that, for any $t\in I$, the operator $H(t)$ is densely defined on ${\D(H(t))}$.
		We say that $H(t)$,  ${t \in I}$, is a \emph{time-dependent Hamiltonian with constant form domain} $\mathcal{H}^+$ if:
		\begin{enumerate}[label=\textit{(\roman*)},nosep,leftmargin=*]
			\item There is $m >0$ such that, for any $t\in I$, $\langle \Phi, H(t)\Phi \rangle \geq -m \|\Phi\|^2$ for all $\Phi\in{\D(H(t))}$.
			\item For any $t \in I$, the domain of the Hermitian sesquilinear form $h_t$ associated with $H(t)$ by the Representation Theorem is $\mathcal{H}^+$.
		\end{enumerate}
	\end{definition}  
	Hamiltonians with constant form domain allow us to define some nested chains (\textit{scales}) of Hilbert spaces. These are the key objects that we will use to prove the stability result and its applications to quantum control.
	
	\begin{definition} \label{def:scales_hilbert}
		Let $I \subset \mathbb{R}$ be an interval, let $\mathcal{H}^+$ be a dense subspace of $\mathcal{H}$ and ${H(t)}$, ${t \in I}$, a time-dependent Hamiltonian with constant form domain $\H^+$.
		The \emph{scale of Hilbert spaces} defined by $H(t)$, cf.\ \cite[Section 1.1]{Berezanskii1968}, is the triple of Hilbert spaces
		\begin{equation*}
			\H^+_t:=(\mathcal{H}^+, \langle \cdot, \cdot \rangle_{+,t}) \subset
			(\mathcal{H}, \langle \cdot, \cdot \rangle) \subset
			(\mathcal{H}^-_t, \langle \cdot, \cdot \rangle_{-,t}),
		\end{equation*}
		where $\langle \Psi, \Phi \rangle_{\pm,t} := \langle (H(t) + m + 1)^{\pm\sfrac{1}{2}}\Psi, (H(t) + m + 1)^{\pm\sfrac{1}{2}} \Phi \rangle$, and $\mathcal{H}^-_t$ denotes the closure of $\mathcal{H}$ with respect to the norm defined by $\|\Phi\|_{-,t}^2 \coloneqq \langle \Phi, \Phi \rangle_{-,t}$. We will denote by $(\argdot,\argdot) : \H^+_t\times\H^-_t \to \mathbb{C}$ the canonical pairings.   
	\end{definition}
	
	\begin{remark}
	$\phantom{V}$
	\begin{enumerate}[label=\textit{(\roman*)}]
           \item An important property of the scales of Hilbert spaces, which will be used extensively, is the following one:
	\begin{equation}\label{eq:norm_hierarchy}
\norm{\argdot}_{-,t}\leq \norm{\argdot}\leq \norm{\argdot}_{+,t}.
	\end{equation}
	\item Notice that the canonical pairing is a continuous extension of the scalar product and it does not depend on $t$.
	\end{enumerate}
	\end{remark}
	\begin{remark}\label{rem:generic_scale}
	The above definition of scales of Hilbert spaces is canonically associated with the particular Hamiltonian with time-dependent domain. However, whenever one has a densely and continuously embedded Hilbert space $i:\H^+\to\H$, one can canonically construct a scale of Hilbert spaces associated with it \cite{Berezanskii1968}.
	\end{remark}
		
	\begin{definition}\label{def:unitarypropagator} A \emph{unitary propagator} is a two-parameter family of unitary operators $U(t,s)$, $s,t\in\R$ that satisfies:
		\begin{enumerate}[label=\textit{(\roman*)},nosep,leftmargin=*]
			\item $U(t,s)U(s,r)= U( t , r)$ {for all $t,s,r\in \R$};
			\item $U(t,t)= \mathbb{I}$ {for all $t\in \R$};
			\item ${t,s\mapsto\,}U(t,s)$ is jointly strongly continuous in $t$ and $s$.
		\end{enumerate}
	\end{definition}
	Next we introduce several notions of solutions of the Schrödinger equation. 
	
	\begin{definition}\label{def:strong_Schrodinger}
		Let $I\subset\R$ be an interval and $H(t)$, $t\in I$, be a time-dependent Hamiltonian. 
		Consider the equation
		\begin{equation}\label{eq:schro0}
			i\frac{\mathrm{d}}{\mathrm{d}t}\Psi(t)=H(t)\Psi(t).
		\end{equation}
		We say that a unitary propagator $U(t,s)$, $t,s\in I$, is a \textit{(strong) solution of the Schr\"odinger equation for $H(t)$} if, for all $\Psi_0\in\mathcal{D}(H(s))$, the function $t\in I\mapsto\Psi(t):=U(t,s)\Psi_0$ solves Eq.~\eqref{eq:schro0} with initial condition $\Psi(s)=\Psi_0$.
	\end{definition}
	
	\begin{definition}\label{def:weak_Schrodinger}
		Let $I\subset\R$ be an interval and $H(t)$, $t\in I$, be a time-dependent Hamiltonian with constant form domain $\H^+$,
		and let $\Phi\in\H^+$. Consider the equation	
		\begin{equation}\label{eq:schro_weak}
			i\frac{\mathrm{d}}{\mathrm{d}t}\scalar{\Phi}{\Psi(t)}=h_t\left(\Phi,\Psi(t)\right),
		\end{equation}
		with $h_t(\cdot,\cdot)$ being the sesquilinear form uniquely associated with $H(t)$. We say that a unitary propagator $U(t,s)$, $t,s\in I$, is a \textit{weak solution of the Schr\"odinger equation  for $H(t)$} if, for $\Psi_0\in\H^+$, the function $t\in I\mapsto \Psi(t):=U(t,s)\Psi_0$ solves Eq.~\eqref{eq:schro_weak} with initial condition $\Psi(s)=\Psi_0$ for all $\Phi\in\H^+$.
	\end{definition}
	
	For quantum control purposes, one often needs an even weaker notion of solution that takes into account propagators solving Eq.~\eqref{eq:schro_weak} for all but finitely many values of $t$, that is, admitting finitely many time singularities, see e.g.~\cite{BoscainCaponigroChambrionEtAl2012,boussaid2013weakly,ChambrionMasonSigalottiEtAl2009}. This leads us to the following definition:
	\begin{definition}\label{def:admissible}
		Let $I\subset\R$ be an interval and $H(t)$, $t\in I$, be a time-dependent Hamiltonian with constant form domain $\H^+$. We say that a unitary propagator $U(t,s)$, $t,s\in I$, is a \textit{piecewise weak solution of the Schr\"odinger equation for $H(t)$} if, for all $\Psi_0\in\H^+$, there exist $t_0<t_1<\ldots<t_d\in I$ and a family of weak solutions of the Schr\"odinger equation $\{U_i(t,s) \mid t,s \in (t_{i-1},t_i)\}_{i=1,\dots,d}$ such that, for $t\in(t_{i-1},t_i)$, $s\in(t_{j-1},t_j)$, $1\leq j<i \leq d$, the unitary propagator $U(t,s)$ can be expressed as
		\begin{equation}
			U(t,s)=U_i(t,t_{i-1})U_{i-1}(t_{i-1},t_{i-2})\cdots U_j(t_{j},s).
		\end{equation}
	\end{definition}
	
	It is straightforward to prove that strong solutions of the Schrödinger equation are also weak solutions, and that the latter are also piecewise weak solutions. 
	
	Several results of this article are obtained for families of functions that are piecewise differentiable. The precise definition of these sets of functions is given next:
	
	\begin{definition} \label{def:piecewise_differentiable}
		We will say that a function $f$ is $n$-times \textit{piecewise differentiable} on $I \subset \mathbb{R}$, denoted $f \in C_{\mathrm{pw}}^n(I)$, if there exists a finite collection $\{I_j\}_{j=1}^\nu$ of pairwise disjoint open subintervals of $I$ such that $I = \bigcup_{j} \overline{I_j}$ and there exists a collection of $n$-times differentiable functions $\{g_j\}_{j=1}^\nu \subset C^n(I)$ satisfying $f|_{I_j} = g_j|_{I_j}$ for $j = 1, \dots, \nu$.
	\end{definition}
	Note that this definition guarantees that, for each interval $I_j$, the function and its derivatives can be continuously extended to $\overline{I_j}$.
	The strong or weak solutions of the  Schrödinger equation with initial value $\Psi$ at time $s\in\R$ are given, if they exist, by the curves $\Psi(t) := U(t,s)\Psi$. We leave the analysis on the existence of such solutions for Section~\ref{sec:form_dynamics}. 	
		
	
	
	
	\section{Hamiltonian forms and quantum dynamics} \label{sec:form_dynamics}
	
	\subsection{Stability} \label{subsec:stability}
	For time-dependent Hamiltonians with constant form domain, conditions for the existence of solutions of the Schrödinger equation can be given.
	The following theorem is proven, in two slightly different forms, in \cite[Sec. 8]{Kisynski1964} and~\cite[Thms. II.23 and II.24]{Simon1971}. We refer also to~\cite{balmaseda2022schrodinger} for a detailed discussion on the relation between these approaches.
	We will assume for the rest of this article that $I$ is a compact interval. This is a necessary requirement to obtain some of the results.
	
	\begin{theorem}[\cite{Kisynski1964} and \cite{Simon1971}] \label{thm:kisynski}
		Let $H(t)$, ${t \in I}$, be a time-dependent Hamiltonian with constant form domain $\mathcal{H}^+$
{and, for all $t\in I$, dense domain $\D(H(t))\subset\H$.}
		Suppose that the sesquilinear form associated with $H(t)$, $h_t$, is such that for any $\Phi, \Psi \in \mathcal{H}^+$ the map $t \mapsto h_t(\Phi, \Psi)$ is in $C^1(I)$.
		Then there exists a unitary propagator $U(t, s)$, $t,s\in I$, such that:
		\begin{enumerate}[label=\textit{(\roman*)},nosep,leftmargin=*]
			\item\label{thm:kisynski_i} $U(t, s) \mathcal{H}^+ = \mathcal{H}^+$, $t,s\in I$.
			\item\label{thm:kisynski_ii} For every $\Phi,\Psi \in \mathcal{H}^+$ the function $t \mapsto \langle \Phi, U(t,s)\Psi \rangle$ is continuously differentiable and $U(t,s)$ is a weak solution of the Schrödinger equation. 
		\end{enumerate}
		If, moreover, for every $\Phi,\Psi \in \mathcal{H}^+$ the function $t \mapsto h_t(\Phi, \Psi)$ is in $C^2(I)$, then the following properties also hold:
		\begin{enumerate}[resume*]
			\item\label{thm:kisynski_iii} $U(t, s) {\mathcal{D}(H(s)) = \mathcal{D}(H(t))}$.
			\item\label{thm:kisynski_iv} For every $\Psi \in {\mathcal{D}(H(s))}$, the function $t \mapsto U(t,s)\Psi$ is continuously differentiable and $U(t,s)$ is a strong solution of the Schrödinger equation.
		\end{enumerate}
	\end{theorem}
	
	This theorem ensures the existence of weak solutions  whenever the form associated with the Hamiltonian is $C^1$ and strong solutions of the Schrödinger equation in the case that it is $C^2$.
	The following result is an extension to the case of piecewise continuous functions.
	
	\begin{corollary} \label{corol:dynamics_piecewise}
		Let $H(t)$, ${t \in I}$, be a time-dependent Hamiltonian with constant form domain $\mathcal{H}^+$ such that $t \mapsto h_t(\Psi,\Phi) \in C_{\mathrm{pw}}^1(I)$ for every $\Phi,\Psi\in\H^+$.
		Suppose that the collection of intervals $\{I_j\}_{j=1}^\nu$ such that the restrictions $t \in I_j \mapsto h_t(\Phi, \Psi) \in \mathbb{C}$ are in $C^1(I_j)$ can be chosen independently of $\Psi$ and $\Phi$.
		Then there exists a strongly continuous unitary propagator that satisfies:
		\begin{enumerate}[label=\textit{(\roman*)},nosep,leftmargin=*]
		  \item\label{cor:kisynski_i} $U(t, s) \mathcal{H}^+ = \mathcal{H}^+$, $t,s\in I$.
		  \item\label{cor:kisynski_ii} For every $\Phi,\Psi \in \mathcal{H}^+$ the function $t \mapsto \langle \Phi, U(t,s)\Psi \rangle$ is continuously differentiable except at a finite set of points, and $U(t,s)$ is a piecewise weak solution of the Schrödinger equation.
		\end{enumerate}
		
	\end{corollary}
	\begin{proof}
		By Theorem~\ref{thm:kisynski}, for every $j$ there exists a weak solution $U_j(t,s)$ of the Schrödinger equation on $\overline{I_j}$.
		Let $t_1 < t_2 < \ldots < t_{\nu+1}$ be such that $I_j = (t_j, t_{j+1})$.
		
		Let $s,t \in I$ with $s < t$.
		If $s,t \in I_r$ for some $r$, define $U(t,s) = U_r(t,s)$.
		If $s \in I_r$ and $t \in I_j$ for $r \neq j$ define
		\begin{equation*}
			U(t, s) \coloneqq U_{j}(t, t_{j}) U_{j-1}(t_{j}, t_{j-1}) \cdots U_{r+1}(t_{r+1}, t_r) U_{r}(t_r, s).
		\end{equation*}
		The operator $U(t,s)$ is a strongly continuous unitary propagator, cf.\ Definition~\ref{def:unitarypropagator}.
		Since the form domain $\H^+$ of $H(t)$ is constant and is preserved by $U(t,s)$, cf.\ Theorem~\ref{thm:kisynski}, $U(t,s)$ is a piecewise weak solution of the Schrödinger equation.
	\end{proof}
	
	\begin{remark}
		\phantom{v}
			Since $t \mapsto h_t(\Psi,\Phi)$ is not in $C^2(I)$, the propagator $U(t,s)$ defined above is not a strong solution of the Schrödinger equation with Hamiltonian $H(t)$, for $U(t,s) \Psi$ might not be in $\D(H(t))$, even if $\Psi \in \D(H(s))$, unless $t$ and $s$ lie in the same open interval $I_{j}$.
	\end{remark}
	
		{The following assumptions determine the class of problems for which we will obtain the stability result, Theorem~\ref{thm:stability_bound}:
	\begin{assumption} \label{assump:stability_assumptions}
		Let $I \subset \mathbb{R}$ be a compact interval,
		$\mathbf{N}\subset\mathbb{N}$, and 
		{let $\{H_n(t)\}_{n\in\mathbf{N}}$, ${t \in I}$, be a family of time-dependent Hamiltonians with constant form domain $\mathcal{H}^+$.}	
		We assume:
		\begin{enumerate}[label=\textit{(A\arabic*)},nosep,leftmargin=*]
			\item \label{assump:uniformbound}
			There is $m > 0$ such that $H_n(t) > -m$ for every $n \in \mathbf{N}$ and every $t\in I$.
			\item \label{assump:differentiability_forms}
			For each $n \in \boldsymbol{N}$ and every $\Phi,\Psi$ in $\mathcal{H}^+$, $t \mapsto h_{n,t}(\Psi,\Phi)$ is in $C_{\mathrm{pw}}^2(I)$ and the collection of open subintervals $\{I_j\}_{j=1}^\nu$ on which the restrictions of $t \mapsto h_{n,t}(\Psi, \Phi)$ is $C^2$ can be chosen independently of $\Psi$ and $\Phi$.
			\item \label{assump:norm_equiv}
			Given $n_0\in \mathbf{N}$ and $t_0\in I$, there is $c > 0$ such that, for all $n\in\mathbf{N}$ and $t\in I$, the norms of the corresponding scales of Hilbert spaces satisfy
			\begin{equation*}
				c^{-1} \|\cdot\|_{\pm,n,t} \leq \|\cdot\|_{\pm,n_0,t_0} \leq c \|\cdot\|_{\pm,n,t}.
			\end{equation*}
			\item \label{assump:L1_deriv_op_bound}
			The family of real functions on $\bigcup_j I_j$ defined by
			\begin{equation*}
				\tilde{C}_n(t) \coloneqq \sup_{\substack{\Psi, \Phi \in \mathcal{H}^+ \\ \|\Psi\|_+ = \|\Phi\|_+ = 1}} \left|\frac{\d}{\d t} h_{n,t}(\Psi, \Phi)\right|
			\end{equation*}
			satisfies 
				$M \coloneqq \sup_n \sum_{j=1}^\nu \|\tilde{C}_n\|_{L^1(I_j)} <\infty.$
		\end{enumerate}
	\end{assumption}
	\begin{remark} \label{remark:stability_assump}
		$\phantom{V}$
		\begin{enumerate}[label=\textit{(\roman*)},nosep,leftmargin=*]
			\item We keep the set $\mathbf{N}$ general since we want to cover the cases in which $\mathbf{N}$ is a finite set or $\mathbf{N}=\mathbb{N}$.
			\item The uniform equivalence of the norms, \ref{assump:norm_equiv}, motivates the use of the simplified notation $\|\Phi\|_\pm \coloneqq \|\Phi\|_{\pm,n_0,t_0}$, {for an arbitrary choice of $n_0\in\mathbf{N}$ and $t_0\in I$,} that we will use throughout the text. We will consider $\mathcal{H}^\pm$ as Hilbert spaces endowed with the norms $\|\cdot\|_\pm$.
			\item It can be shown, cf.\ \cite[Proposition~1]{BalmasedaLonigroPerezPardo2022}, that for each $n \in \boldsymbol{N}$,~\ref{assump:differentiability_forms} implies that $\tilde{C}_n(t)$ is in $C_{\mathrm{pw}}^1(I)$.
			Since $I$ is compact, it is therefore guaranteed that the restrictions $\tilde{C}_n|_{I_j}$ are in $L^1(I_j)$.
		\end{enumerate}
	\end{remark}
	}

	Before proving the stability theorem, Theorem~\ref{thm:stability_bound}, we shall need some preliminary results.	
	Any Hermitian sesquilinear form $v$ densely defined on $\H^+$, cf.\ Definition~\ref{def:scales_hilbert}, which is bounded with respect to the norm $\|\argdot\|_+$, can be interpreted, using the scale of Hilbert spaces and the canonical pairing, as a bounded operator $V: \mathcal{H}^+ \to \mathcal{H}^-$ that satisfies
	\begin{equation*}
		(\Psi, V\Phi) = v(\Psi, \Phi), \qquad \Phi, \Psi \in \mathcal{H}^+.
	\end{equation*}
	Its norm is given by
	\begin{equation}\label{eq:operator_norm_form}
		\|V\|_{+,-} = \sup_{\Psi, \Phi \in \mathcal{H}^+ \smallsetminus \{0\}} \frac{|v(\Psi, \Phi)|}{\|\Psi\|_+\|\Phi\|_+}.
	\end{equation}
	{Therefore, a family of unbounded operators $H_n(t)$ densely defined on ${\D(H_n(t))}$ as in Assumption~\ref{assump:stability_assumptions} and with constant form domain $\H^+$ can be continuously extended to bounded operators $\tilde{H}_n(t): \mathcal{H}^+ \to \mathcal{H}^-$.
	Analogously, if the inverse $H_n(t)^{-1}$ exists as a bounded operator on $\H$, it can be extended to a bounded operator from $\mathcal{H}^-$ to $\mathcal{H}^+$, which coincides with $\tilde{H}_n(t)^{-1}$, cf.\ \cite{Berezanskii1968}.}
	To simplify the notation, we will hereafter denote the original operators and their extensions by the same symbols. In what follows we will use $\norm{\cdot}_{+,-}$ to denote the operator norm in $\mathcal{B}(\H^+,\H^-)$ and similarly for $\norm{\cdot}_{-,+}$.
	
	\begin{lemma} \label{lemma:Ainverse_convergence}
		Let $\{H_n(t)\}_{n\in \mathbf{N}}$, $t\in I$, be a family of time-dependent Hamiltonians with constant form domain $\H^+$ that satisfy~\ref{assump:uniformbound} and~\ref{assump:norm_equiv}.
		Denote $A_n(t) = H_n(t) + m + 1$. For any $j,k\in \mathbf{N}$ the following inequality holds:
		\begin{equation*}
			\|A_j(t)^{-1} - A_k(t)^{-1}\|_{-,+} \leq c^4 \|H_j(t) - H_k(t)\|_{+,-}.
		\end{equation*}
	\end{lemma}
	\begin{proof}
		By definition, for any $n\in \mathbf{N}$ we have $\|\Phi\|_{\pm,n,t} = \|A_n(t)^{\pm \sfrac{1}{2}}\Phi\|$.
		By~\ref{assump:norm_equiv} we have
		\begin{equation*}
			\|A_n(t)^{-1}\Phi\|_+ \leq c \|A_n(t)^{-1}\Phi\|_{+,n,t} = c \|\Phi\|_{-,n,t} \leq c^2 \|\Phi\|_-,
		\end{equation*}
		and therefore $\|A_n(t)^{-1}\|_{-,+} \leq c^2$.
		Hence it follows 
		\begin{align*}
			\|A_j(t)^{-1} - A_k(t)^{-1}\|_{-,+} &= \|A_j(t)^{-1} [A_k(t) - A_j(t)] A_k(t)^{-1}\|_{-,+} \\
			&\leq \|A_j(t)^{-1}\|_{-,+} \|A_j(t) - A_k(t)\|_{+,-} \|A_k(t)^{-1}\|_{-,+} \\
			&\leq c^4 \|H_j(t) - H_k(t)\|_{+,-}.\qedhere
		\end{align*}
	\end{proof}
	
	\begin{lemma} \label{lemma:bound_derivative_operators}
		Let $\{H_n(t)\}_{n\in \mathbf{N}}$, $ t\in I$, be a family of time-dependent Hamiltonians with constant form domain that satisfies Assumption~\ref{assump:stability_assumptions}.
		For each $n$, let $\{I_{n,j}\}_{j=1}^{\nu_n}$ be the open intervals where $t \mapsto h_{t,n}(\Psi,\Phi)$ is differentiable (cf.\ \ref{assump:differentiability_forms}).
		Denote $A_n(t) = H_n(t) + m + 1$, where $m>0$ is the uniform lower bound.
		Then for any $t \in \bigcup_{j} I_{n,j}$, the derivative $\frac{\d}{\d t} \left(A_n(t)^{-1}\right)$ exists in the sense of the norm of $\mathcal{B}(\mathcal{H})$, and the functions
		\begin{equation*}
			C_n(t) \coloneqq \left\|A_n(t)^{\sfrac{1}{2}} \frac{\d}{\d t} \left(A_n(t)^{-1}\right) A_n(t)^{\sfrac{1}{2}}\right\|
		\end{equation*}
		satisfy
		\begin{equation*}
			\sum_{j=1}^{\nu_n} \int_{I_{n,j}} C_n(t) \, \d t \leq c^2 M,
		\end{equation*}
		with $M$ defined in~\ref{assump:L1_deriv_op_bound}.
	\end{lemma}
	\begin{proof}
		From~\ref{assump:differentiability_forms} and~\cite[Proposition~2]{balmaseda2022schrodinger} it follows that $\frac{\d}{\d t} A_n(t)$ 
		exists in the sense of $\mathcal{B}(\mathcal{H^+,\H^-})$
		everywhere but at the boundaries of the intervals $I_{n,j}$, which is a finite set. Moreover, it is the operator associated with the sesquilinear form $\frac{\d}{\d t} h_{n,t}(\Psi, \Phi)$.
		
		Notice that, for every $n\in \mathbf{N}$ and $t\in I$, the operators $A_n(t)^{\sfrac{1}{2}}\colon\H^+\to\H$ are bijections that can be extended to operators from $\H$ to $\H^-$ that we will denote with the same 
		symbol.
		These operators have bounded inverses $A_n(t)^{-\sfrac{1}{2}}$.
		Let $A_0 \coloneqq A_{n_0}(t_0)$ so that $\norm{\Phi}_{\pm} = \Vert A_0^{\pm\sfrac{1}{2}}\Phi \Vert$. 
		Let $t \in \bigcup_j I_{n,j}$ and $G_p(t) \coloneqq  p^{-1}\left[A_n(t+p) - A_n(t)\right] - \frac{\d}{\d t}A_n(t)$. We have 
		\begin{equation*}
			\|A_0^{-\sfrac{1}{2}} G_p(t) A_0^{-\sfrac{1}{2}}\| = \sup_{\Psi \in \mathcal{H} \smallsetminus \{0\}} \frac{\|A_0^{-\sfrac{1}{2}} G_p(t) A_0^{-\sfrac{1}{2}}\Psi\|}{\|\Psi\|}
			= \sup_{\Phi \in \mathcal{H}^+ \smallsetminus \{0\}} \frac{\|G_p(t)\Phi\|_-}{\|\Phi\|_+},
		\end{equation*}
		where we have used the definitions of the norms $\norm{\argdot}_{\pm}$, and 
		{the property~\eqref{eq:norm_hierarchy}.}
		This shows that $\frac{\d}{\d t} A_0^{-\sfrac{1}{2}} A_n(t) A_0^{-\sfrac{1}{2}}$ exists in the sense of the norm of $\mathcal{B}(\mathcal{H})$. By the product rule and the bijectivity of $A_n(t)$ this implies that $\frac{\d}{\d t} A_0^{\sfrac{1}{2}} A_n(t)^{-1} A_0^{\sfrac{1}{2}}$
		also exists in the sense of the norm of $\mathcal{B}(\mathcal{H})$.
		
		Let $T_p(t)\coloneqq \H^-\to\H^+$ be the operator defined by
		$$ T_p(t) \coloneqq p^{-1}[A_n(t+p)^{-1} - A_n(t)^{-1}] - A_0^{-\sfrac{1}{2}}\frac{\d}{\d t} \left(A_0^{\sfrac{1}{2}} A_n(t)^{-1} A_0^{\sfrac{1}{2}}\right)A_0^{-\sfrac{1}{2}}.$$
		It follows that
		\begin{align*}
			\|T_p(t)|_{\mathcal{H}}\| = \sup_{\Psi \in \mathcal{H}\smallsetminus \{0\}} \frac{\|T_p(t)\Psi\|}{\|\Psi\|}
			&\leq \sup_{\Psi \in \mathcal{H}\smallsetminus \{0\}} \frac{\|T_p(t)\Psi\|_+}{\|\Psi\|_-}= \sup_{\Psi \in \mathcal{H}\smallsetminus \{0\}} \frac{\|A_0^{\sfrac{1}{2}}T_p(t)\Psi\|}{\|\Psi\|_-}\\
			&= \sup_{\Phi \in \mathcal{H^+}\smallsetminus \{0\}} \frac{\|A_0^{\sfrac{1}{2}}T_p(t)A_0^{\sfrac{1}{2}}\Phi\|}{\|\Phi\|_+} \\&\leq \norm{A_0^{\sfrac{1}{2}}T_p(t)A_0^{\sfrac{1}{2}}}.
		\end{align*}
		But $\lim_{p\to0}\norm{A_0^{\sfrac{1}{2}}T_p(t)A_0^{\sfrac{1}{2}}}=0$, and hence $\frac{\d}{\d t} \left(A_n(t)^{-1}\right)$ exists in the sense of the norm of $\mathcal{B}(\mathcal{H})$.
		
		As for the second statement, by the product rule we have
		\begin{equation*}
			\frac{\d}{\d t} \left(A_n(t)^{-1}\right) = - A_n(t)^{-1} \frac{\d}{\d t} A_n(t) A_n(t)^{-1},
		\end{equation*}
		and therefore
		\begin{align*}
			\norm{A_n(t)^{\sfrac{1}{2}} \frac{\d}{\d t}\left( A_n(t)^{-1}\right) A_n(t)^{\sfrac{1}{2}}} &= \norm{A_n(t)^{-\sfrac{1}{2}} \frac{\d}{\d t} \left(A_n(t)\right) A_n(t)^{-\sfrac{1}{2}}}\\
			&\leq c \sup_{\Psi \in \mathcal{H}\smallsetminus \{0\}}  \frac{\norm{\frac{\d}{\d t} A_n(t) A_n(t)^{-\sfrac{1}{2}}\Psi}_-}{\norm{\Psi}}\\
			&\leq c \sup_{\Psi \in \mathcal{H}\smallsetminus \{0\}}  \norm{\frac{\d}{\d t} A_n(t) }_{+,-}\frac{\norm{A_n(t)^{-\sfrac{1}{2}}\Psi}_+}{\norm{\Psi}}\\
			&\leq c^2 \norm{\frac{\d}{\d t} A_n(t) }_{+,-}.
		\end{align*}
		Using~\ref{assump:L1_deriv_op_bound} and Eq.\ \eqref{eq:operator_norm_form} the statement follows.
	\end{proof}
	
	\begin{lemma} \label{lemma:propagators_bound_simon}
		Let $\{H_n(t)\}_{n\in \mathbf{N}}$, $ t\in I$, be a family of time-dependent Hamiltonians with constant form domain that satisfies Assumption~\ref{assump:stability_assumptions}. Define $A_n(t) := H_n(t) + m + 1$ and let $\Phi_n(t) = U_n(t,s)\Phi$, where $\Phi \in \mathcal{H}^+$ and $U_n(t,s)$ is the weak solution of the Schrödinger equation for $A_n(t)$.
		For every $n$ the following bounds hold
		\begin{gather*}
			\|A_n(t)\Phi_n(t)\|_{-,n,t} =
			\|\Phi_n(t)\|_{+,n,t} \leq e^{\frac{1}{2}\int_s^t C_n(\tau)\,\d\tau}\|\Phi\|_{+,n,s}, \\
			\|\Phi_n(t)\|_{-,n,t} \leq e^{\frac{1}{2}\int_s^t C_n(\tau)\,\d\tau}\|\Phi\|_{-,n,s},
		\end{gather*}
		with $C_n(t) := \norm{A_n(t)^{\sfrac{1}{2}} \frac{\d}{\d t} \left(A_n(t)^{-1}\right) A_n(t)^{\sfrac{1}{2}}}$.
	\end{lemma}
	\begin{proof}
		For each $n \in \boldsymbol{N}$, let $\{I_{n,j}\}_{j=1}^{\nu_n}$ be the open intervals where $t \mapsto h_{t,n}(\Psi,\Phi)$ is differentiable for each $\Psi, \Phi \in \mathcal{H}^+$.
    We will show the bounds for $t,s$ in the same interval $I_{n,j}$, since the extension to the whole interval $I$ follows straightforwardly from the additivity of the integral.
    
    For $t \in I_{n,j}$, let $\mathcal{D}_t \coloneqq \{\Phi \in \mathcal{H}^+ \mid A_n(t)\Phi \in \mathcal{H}^+\}$.
    By \cite[Lemma~7.4]{Kisynski1964} (see also the proof of \cite[Lemma~7.9]{Kisynski1964}), $\mathcal{D}_t$ is a dense subset of $\mathcal{H}^+$.
    By \cite[Theorem~8.1]{Kisynski1964}, for any $t,s \in I_{n,j}$ and $\Phi \in \mathcal{D}_s$, $\Phi_n(t) = U_n(t,s)\Phi$ is in $\mathcal{D}_t$ and the derivative of $t \mapsto \Phi_n(t)$ exists both in the sense of $\mathcal{H}^-$ and $\mathcal{H}^+$.
    Moreover, in both senses it is equal to $\dot{\Phi}_n(t) \coloneqq -i A_n(t) \Phi_n(t)$.

    From \cite[Proposition 2]{balmaseda2022schrodinger} and Lemma~\ref{lemma:bound_derivative_operators} we have that $\dot{A}_n(t):=\frac{\d}{\d t}{A}_n(t)\in \mathcal{B}(\mathcal{H^+,\H^-})$ and $B_n(t):=\frac{\d}{\d t}{A}_n(t)^{-1} \in \mathcal{B}(\mathcal{H^-,\H^+})$.
    %
    By the product rule, it follows that $\dot{A}_n(t) A_n(t)^{-1} = - A_n(t) B_n(t)$.
    Let us calculate now the derivatives of $t \in I \mapsto \dot{\Phi}_n(t) \in \mathcal{H}_-$ and $t \in I \mapsto A_n(t)^{-1} \dot{\Phi}_n(t) \in \mathcal{H}^+:$
    \begin{gather*}
      \frac{\d}{\d t} \dot{\Phi}_n(t)
      = -i \dot{A}_n(t) \Phi(t) - iA_n(t) \dot{\Phi}_n(t)
      = - A_n(t) B_n(t) \dot{\Phi}_n(t) - iA_n(t) \dot{\Phi}_n(t), \\
      \frac{\d}{\d t} \left(A_n(t)^{-1} \dot{\Phi}_n(t)\right) = B_n(t) \dot{\Phi}_n(t) + A_n(t)^{-1} \frac{\d}{\d t} \dot{\Phi}_n(t)
      =  - i \dot{\Phi}(t).
    \end{gather*}

    Let $y_n(t) \coloneqq \|A_n(t)\Phi_n(t)\|_{-,n,t}^2 = \|\dot{\Phi}_n(t)\|_{-,n,t}^2$.
    Then, one has
    \begin{equation*}
      y_n(t) = \langle A_n(t)^{-\sfrac{1}{2}}\dot{\Phi}_n(t), A_n(t)^{-\sfrac{1}{2}}\dot{\Phi}_n(t) \rangle
      = \left( \dot{\Phi}_n(t), A_n(t)^{-1} \dot{\Phi}_n(t) \right),
    \end{equation*}
    where $(\cdot, \cdot)$ denotes the canonical pairing between $\mathcal{H}^+$ and $\mathcal{H}^-$.
    Having into account that $\Phi_n(t) = U_n(t,s)\Phi \in \mathcal{D}_t$ and $\dot{\Phi}_n(t) = -i A_n(t)\Phi_n(t)\in \H^+$ we get
    \begin{align*}
      \frac{\d}{\d t} y_n(t)
      &= \left( \frac{\d}{\d t} \dot{\Phi}_n(t), A_n(t)^{-1} \dot{\Phi}_n(t) \right) + \left( \dot{\Phi}_n(t), \frac{\d}{\d t} \left(A_n(t)^{-1} \dot{\Phi}_n(t)\right) \right) \\
      &= \left( -B_n(t) \dot{\Phi}_n(t) - i \dot{\Phi}_n(t), \dot{\Phi}_n(t)\right)
      + \left(\dot{\Phi}_n(t),- i \dot{\Phi}_n(t)\right) \\
      &= \left( - B_n(t) \dot{\Phi}_n(t),  \dot{\Phi}_n(t)\right)
    \end{align*}
    Since 
    \begin{align*}
      \left|\left( \dot{\Phi}_n(t), B_n(t) \dot{\Phi}_n(t) \right)\right|
      &= \left|\left\langle A_n(t)^{-\sfrac{1}{2}} \dot{\Phi}_n(t), A_n(t)^{\sfrac{1}{2}} B_n(t) A_n(t)^{\sfrac{1}{2}} A_n(t)^{-\sfrac{1}{2}} \dot{\Phi}_n(t) \right\rangle\right| \\
      &\leq \|A_n(t)^{\sfrac{1}{2}} B_n(t) A_n(t)^{\sfrac{1}{2}}\| \|\dot{\Phi}_n(t)\|_{-,n,t}^2,
    \end{align*}
    it follows
    \begin{equation*}
      \frac{\d}{\d t} y_n(t) \leq \|A_n(t)^{\sfrac{1}{2}} B_n(t) A_n(t)^{\sfrac{1}{2}}\| y_n(t).
    \end{equation*}
    Noting that $\|A_n(t) \Phi_n(t)\|_{-,n,t} = \|\Phi_n(t)\|_{+,n,t}$, the first inequality follows now from Grönwall's inequality for any $\Phi \in \mathcal{D}_s$.
    However, since $\mathcal{D}_s$ is dense in $\mathcal{H}^+$, the bound can be extended to any $\Phi \in \mathcal{H}^+$.
    The second inequality in the statement can be obtained following a similar argument using Grönwall's inequality with $z(t)=\norm{\Phi_n(t)}^2_{-,n,t}$.
	\end{proof}

	The following theorem is a preliminary result towards our main stability theorem, Theorem~\ref{thm:stability_bound}:
	
	\begin{theorem} \label{thm:stability_bound_c2}
		Let $I \subset \mathbb{R}$ be a compact interval, $\mathbf{N}\subset \mathbb{N}$, and $\{H_n(t)\}_{n\in \mathbf{N}}$, $t\in I$, be a family of time-dependent Hamiltonians with constant form domain $\mathcal{H}^+$ that satisfies assumptions~\ref{assump:uniformbound},~\ref{assump:norm_equiv},~\ref{assump:L1_deriv_op_bound} and such that, for every $\Phi, \Psi$, the function $t \mapsto h_{n,t}(\Psi,\Phi)$ is in $C^2(I)$.
		Then there is a constant $L$ depending only on $c$ and $M$ such that, for any $j,k\in \mathbf{N}$ and any $t,s \in I$, it holds
		\begin{equation*}
			\|U_j(t,s) - U_k(t,s)\|_{+,-} \leq L \int_s^t \|H_j(\tau) - H_k(\tau)\|_{+,-} \, \d\tau,
		\end{equation*}
		where $U_n(t,s)$ is the weak solution of the Schrödinger equation for $H_n(t)$, $n\in \mathbf{N}$, and $\|\cdot\|_{+,-}$ is the norm in $\mathcal{B}(\mathcal{H}^+,\mathcal{H}^-)$.
	\end{theorem}
	
	\begin{proof}
		Let $A_n(t) := H_n(t) + m + 1$.
		Clearly, for any $\Phi\in\mathcal{H}^+$, one has $\langle \Phi, A_n(t)\Phi \rangle \geq \|\Phi\|^2$ and therefore $\{A_n(t)\}_{n\in\mathbf{N}}$, $t\in I$, defines a family of time-dependent Hamiltonians with constant form domain $\mathcal{H}^+$, which obviously satisfies Assumption~\ref{assump:stability_assumptions}.
		The unitary propagators $U_n(t,s)$ and $\tilde{U}_n(t,s)$ that are the weak solutions of the Schrödinger equations for $A_n(t)$ and $H_n(t)$, respectively, are related through the equation
		\begin{equation*}
			\tilde{U}_{n}(t,s) = U_{n}(t,s) e^{-i(m + 1)(t-s)}
		\end{equation*}
		and therefore it suffices to prove the theorem for $\{A_n(t)\}_{n\in \mathbf{N}}$.
		
		For $\Phi \in \mathcal{H}^+$, define $\Phi_n(t) \coloneqq U_n(t,s) \Phi$.
		By~\ref{assump:norm_equiv}, we have for $j,k\in \mathbf{N}$
		\begin{equation*}
			\|\Phi_j(t) - \Phi_k(t)\|_- \leq c \|\Phi_j(t) - \Phi_k(t)\|_{-,j,t}.
		\end{equation*}
		For convenience, let us denote $z(t) = \|\Phi_j(t) - \Phi_k(t)\|_{-,j,t}$ and $y(t) = z(t)^2$.
		By the definition of $\|\cdot\|_{-,j,t}$ we have
		\begin{equation*}
			y(t) = \langle\Phi_j(t) - \Phi_k(t), A_j(t)^{-1} [\Phi_j(t) - \Phi_k(t)]\rangle,
		\end{equation*}
		and using the chain rule it follows
		\begin{equation} \label{eq:z_derivative}
			\frac{\d}{\d t} z(t) = \frac{1}{2 z(t)} \frac{\d}{\d t} y(t).
		\end{equation}
		By Lemma~\ref{lemma:bound_derivative_operators}, $A_j(t)^{-1}$ is differentiable in the sense of $\mathcal{B}(\mathcal{H})$ and
		\begin{align*}
			\frac{\d}{\d t} y(t) &= \langle\dot{\Phi}_j(t) - \dot{\Phi}_k(t), A_j(t)^{-1} [\Phi_j(t) - \Phi_k(t)]\rangle \\
			&\phantom{=}+ \langle\Phi_j(t) - \Phi_k(t),  B_j(t)[\Phi_j(t) - \Phi_k(t)]\rangle \\
			&\phantom{=}+ \langle\Phi_j(t) - \Phi_k(t), A_j(t)^{-1}[\dot{\Phi}_j(t) - \dot{\Phi}_k(t)]\rangle,
		\end{align*}
		where we have denoted $\dot{\Phi}_n(t) = \frac{\d}{\d t} \Phi_n(t)$ and $B_n(t) := \frac{\d}{\d t} A_n(t)^{-1}$, $n\in \mathbf{N}$.
		By the self-adjointness of $A_j(t)^{-1}$, we get
		\begin{align} \label{eq:bound_y_derivative}
			\frac{\d}{\d t} y(t) &= 2 \Re\langle\Phi_j(t) - \Phi_k(t), A_j(t)^{-1}[\dot{\Phi}_j(t) - \dot{\Phi}_k(t)]\rangle \\
			&\phantom{=}+\langle\Phi_j(t) - \Phi_k(t),  B_j(t)[\Phi_j(t) - \Phi_k(t)]\rangle. \nonumber
		\end{align}
		Let us bound each of the terms in the equation above.
		On the one hand, we have
		\begin{align*}
			|\langle\Phi_j(t) - \Phi_k(t),  B_j(t)[\Phi_j(t) &- \Phi_k(t)]\rangle| = \\
			&= |\langle A_j(t)^{-\sfrac{1}{2}}[\Phi_j(t) - \Phi_k(t)],  A_j(t)^{\sfrac{1}{2}}B_j(t)A_j(t)^{\sfrac{1}{2}}A_j(t)^{-\sfrac{1}{2}}[\Phi_j(t) - \Phi_k(t)]\rangle| \\
			&\leq \|A_j(t)^{\sfrac{1}{2}}B_j(t)A_j(t)^{\sfrac{1}{2}}\|
			\langle \Phi_j(t) - \Phi_k(t),  A_j(t)^{-1}[\Phi_j(t) - \Phi_k(t)]\rangle \\
			&= C_j(t) z(t)^2,
		\end{align*}
		where $C_j(t) \coloneqq \|A_j(t)^{\sfrac{1}{2}}B_j(t)A_j(t)^{\sfrac{1}{2}}\|$.
		
		On the other hand, using the Schrödinger Equation, we have
			\begin{align*}
				\Re\langle\Phi_j(t) - \Phi_k(t), A_j(t)^{-1}[\dot{\Phi}_j(t) - \dot{\Phi}_k(t)]\rangle&= -\Im\langle\Phi_j(t) - \Phi_k(t), \Phi_j(t) - A_j(t)^{-1}A_k(t)\Phi_k(t)\rangle \\
				&= -\Im \langle \Phi_j(t) - \Phi_k(t), [A_k(t)^{-1} - A_j(t)^{-1}] A_k(t) \Phi_k(t) \rangle
			\end{align*}
		where we have used $\Im\|\Phi_j(t)-\Phi_k(t)\|=0$.
		Noting that $A_j(t)^{-1}$ maps $\mathcal{H}^-$ into $\mathcal{H}^+$, one gets
			\begin{align*}
				\Re\langle\Phi_j(t) - \Phi_k(t), A_j(t)^{-1}[\dot{\Phi}_j(t) - \dot{\Phi}_k(t)]\rangle &\leq  \|\Phi_j(t) - \Phi_k(t)\|_- \|[A_k(t)^{-1} - A_j(t)^{-1}] A_k(t) \Phi_k(t)\|_+ \\
				&\leq  c \,z(t) \|A_k(t)^{-1} - A_j(t)^{-1}\|_{-,+} \|A_k(t) \Phi_k(t)\|_-.
			\end{align*}
		By Assumption~\ref{assump:norm_equiv}, Lemma~\ref{lemma:Ainverse_convergence} and Lemma~\ref{lemma:propagators_bound_simon}, we have
		\begin{equation*}
			\Re\langle\Phi_j(t) - \Phi_k(t), A_j(t)^{-1}[\dot{\Phi}_j(t) - \dot{\Phi}_k(t)]\rangle 
			\leq  c^7 z(t) e^{\frac{1}{2} \int_s^t C_k(\tau)\,\d\tau} \|A_j(t) - A_k(t)\|_{+,-} \|\Phi\|_+.
		\end{equation*}		
		Substituting these results into Equation~\eqref{eq:bound_y_derivative} and using Lemma~\ref{lemma:bound_derivative_operators}, it follows
		\begin{equation*}
			\frac{\d}{\d t} y(t) \leq C_j(t) z(t)^2 + 2c^7 e^{\frac{1}{2}c^2 M} z(t) F(t) \|\Phi\|_+,
		\end{equation*}
		where we have defined $F(t) = \|A_j(t) - A_k(t)\|_{+,-}$.
		By Equation~\eqref{eq:z_derivative},
		\begin{equation*}
			\frac{\d}{\d t} z(t) \leq \frac{C_j(t)}{2} z(t) + c^7 e^{\frac{1}{2}c^2 M} F(t) \|\Phi\|_+.
		\end{equation*}
		Using Grönwall's inequality with initial condition $z(s) = 0$, it follows
		\begin{equation*}
			z(t) \leq c^7e^{c^2M} \|F\|_{L^1(s,t)} \|\Phi\|_+.
		\end{equation*}
		Hence, using again~\ref{assump:norm_equiv},
		\begin{equation*}
			\|\Phi_j(t)-\Phi_k(t)\|_-\leq c^8e^{c^2M}\|F\|_{L^1(s,t)}\|\Phi\|_+.
		\end{equation*}
		Using the definition of the operator norm in $\mathcal{B}(\mathcal{H}^+,\mathcal{H}^-)$ we get
		\begin{equation*}
			\|U_j(t,s) - U_k(t,s)\|_{+,-} \leq c^8e^{c^2M}\|F\|_{L^1(s,t)},
		\end{equation*}
		which completes the proof.
	\end{proof}
	
	Next we present the main stability theorem. 
	This result improves the bounds obtained by A.D.\ Sloan in~\cite[Theorem 9 and Corollary 10]{Sloan1981} and T.\ Kato in~\cite[Theorem III]{Kato1973}. The improvement relies on the fact that this bound can be used for piecewise differentiable generators with respect to an appropriate operator topology. This is not the case in the previously mentioned results. The stability result by Sloan is obtained under the assumption of uniformly bounded derivatives in the appropriate operator topology, cf.~\cite[Eq.~(7)]{Sloan1981}, which implies that the derivative of the generator must be defined everywhere in the interval. In the result by Kato it is required that the time-dependent generator $A(t)$ is continuous in an appropriate operator topology, cf. assumption \cite[Assumption (iii)]{Kato1973}. It is worth to stress that the generalisation for piecewise continuous generators and the bounds of type $L^1$ obtained, both in the derivatives and the generators, 
	 are crucial to prove the controllability results that we present in Section~\ref{sec:applications}.
	
	\begin{theorem} \label{thm:stability_bound}
		Let $I \subset \mathbb{R}$ be a compact interval, $\mathbf{N}\subset \mathbb{N}$ and $\{H_n(t)\}_{n\in \mathbf{N}}$, $t\in I$, be a family of time-dependent Hamiltonians with constant form domain $\mathcal{H}^+$ that satisfies Assumption~\ref{assump:stability_assumptions}.
		Then there is a constant $L$ depending only on $c$ and $M$ such that, for any $j,k\in \mathbf{N}$ and any $t,s \in I$,
		\begin{equation*}
			\|U_j(t,s) - U_k(t,s)\|_{+,-} \leq L \int_s^t \|H_j(\tau) - H_k(\tau)\|_{+,-} \, \d\tau,
		\end{equation*}
		where $U_n(t,s)$ is the weak solution of the Schrödinger equation for $H_n(t)$, $n\in \mathbf{N}$, and $\|\cdot\|_{+,-}$ is the norm in $\mathcal{B}(\mathcal{H}^+,\mathcal{H}^-)$.
	\end{theorem}	
	
	\begin{proof}
		Let $k,\ell \in \boldsymbol{N}$ fixed, and let $\{I_{j}\}_{j=1}^\nu$ be a collection of open intervals satisfying $\bigcup_j \overline{I_j} = I$ such that the functions $t \mapsto h_{i,t}(\Psi, \Phi)$, $\Psi,\Phi \in \mathcal{H}^+$, $i = k,\ell$, are differentiable on each $I_j$ (note that the existence of such collections follows from~\ref{assump:differentiability_forms}).
		Let $t_1 < t_2 < \ldots < t_{\nu+1}$ be such that $I_j = (t_j, t_{j+1})$.
		
		Note that, for $t,s \in I$, Lemma~\ref{lemma:bound_derivative_operators} implies
		\begin{equation} \label{eq:Cn_bound}
			\left|\int_s^t C_n(\tau) \,\d\tau\right| \leq \sum_{j=1}^\nu \|C_n\|_{L^1(I_j)} \leq c^2 M \qquad (n \in \boldsymbol{N}),
		\end{equation}
		where, again, $C_n(t) = \|\frac{\d}{\d t} H_n(t)\|_{+,-,t}$, $t \in \bigcup_j I_j$.
		By Theorem~\ref{thm:stability_bound_c2}, it holds
		\begin{equation} \label{eq:stability_bound}
			\|U_k(t, s) - U_\ell(t, s)\|_{+,-} < L \int_s^t \|H_k(t) - H_\ell(t)\|_{+,-} \,\d\tau,
			\qquad (s,t \in \overline{I_j}).
		\end{equation}
		where the constant $L$ depends only on $c$ and $M$ in Assumption~\ref{assump:stability_assumptions}.
		
		Consider now $0 \leq s < t$ such that $t \in I_{j}$, $s \in I_r$, with $j \neq r$, and denote $\Psi_i(t) \coloneqq U_i(t,s)\Psi$, $i = k,\ell$.
		By Eq.~\eqref{eq:stability_bound}, we have
		\begin{align*}
			\|\Psi_k(t) - \Psi_\ell(t)\|_{-,\ell,t}
			&= \|U_k(t,t_j) \Psi_k(t_j) - U_\ell(t,t_j) \Psi_\ell(t_j)\|_{-,\ell,t} \\
			&\leq \|[U_k(t,t_j) - U_\ell(t,t_j)]\Psi_k(t_j)\|_{-,\ell,t}\;+ \\
			&\phantom{\leq\quad} +  \|U_\ell(t,t_j)[\Psi_k(t_j) - \Psi_\ell(t_j)]\|_{-,\ell,t} \\
			&\leq cL \|F_{k,\ell}\|_{L^1(t_j,t)} \|\Psi_k(t_j)\|_+ \;+\\
			&\phantom{\leq\quad}+ \|U_\ell(t,t_j)[\Psi_k(t_j) - \Psi_\ell(t_j)]\|_{-,\ell,t},
		\end{align*}
		where we have defined $F_{k,\ell}(t) \coloneqq \|H_k(t) - H_\ell(t)\|_{+,-}$.
		
		Let us focus on each of the addends separately.
		First, using the equivalence of the norms~\ref{assump:norm_equiv} and Lemma~\ref{lemma:propagators_bound_simon} one gets
		\begin{equation*}
			\|\Psi_k(t_j)\|_+
			\leq ce^{\frac{1}{2} \int_s^{t_j}C_k(\tau)\,\d\tau} \|\Psi_k(s)\|_{+,k,s}
			\leq c^2 e^{\frac{1}{2} c^2 M} \|\Psi\|_+,
		\end{equation*}
		where we have used Equation~\eqref{eq:Cn_bound}.
		On the other hand, by Lemma~\ref{lemma:propagators_bound_simon} one has
		\begin{equation*}
			\|U_\ell(t,t_j)[\Psi_k(t_j) - \Psi_\ell(t_j)]\|_{-,\ell,t}
			\leq e^{\frac{1}{2}\int_{t_j}^t C_\ell(\tau)\,\d\tau} \|\Psi_k(t_j) - \Psi_\ell(t_j)\|_{-,\ell,t_j}.
		\end{equation*}		
		Therefore, it follows
		\begin{align*}
			\|\Psi_k(t) - \Psi_\ell(t)\|_{-,\ell,t}
			&\leq Lc^3 e^{\frac{1}{2} c^2 M} \|F_{k,\ell}\|_{L^1(t_j,t)} \|\Psi\|_+ +\\
			&\phantom{\leq}+ e^{\frac{1}{2}\int_{t_j}^t C_\ell(\tau)\,\d\tau}\|\Psi_k(t_j) - \Psi_\ell(t_j)\|_{-,\ell,t_j}.
		\end{align*}
		Applying again the same strategy in order to bound $\|\Psi_k(t_j) - \Psi_\ell(t_j)\|_{-,\ell,t_j}$, we get
		\begin{multline*}
			\|\Psi_k(t) - \Psi_\ell(t)\|_{-,\ell,t} \leq Lc^3 e^{\frac{1}{2} c^2 M} \|F_{k,\ell}\|_{L^1(t_j,t)} \|\Psi\|_+ \;+\\
			+ e^{\frac{1}{2}\int_{t_j}^t C_\ell(\tau)\,\d\tau}
			\left[\vphantom{e^{\frac{1}{2}\int_{t_{j-1}}^{t_j} C_\ell(\tau)\,\d\tau}}\right.
			Lc^3 e^{\frac{1}{2} c^2 M} \|F_{k,\ell}\|_{L^1(t_{j-1},t_j)} \|\Psi\|_+ \;+ \\
			\left.+ e^{\frac{1}{2}\int_{t_{j-1}}^{t_j} C_\ell(\tau)\,\d\tau}\|\Psi_k(t_{j-1}) - \Psi_\ell(t_{j-1})\|_{-,\ell,t_{j-1}} \right].
		\end{multline*}
		Using this procedure iteratively and Eq.~\eqref{eq:Cn_bound} we get, for $0 \leq s \leq t$, $s \in I_r$, $t \in I_j$,
		\begin{equation*}
					\|\Psi_k(t) - \Psi_\ell(t)\|_{-,\ell,t} \leq \tilde{L}\left(\|F_{k,\ell}\|_{L^1(t_j,t)} + \|F_{k,\ell}\|_{L^1(t_{j-1},t_j)} + \ldots + \|F_{k,\ell}\|_{L^1(s,t_{r+1})}\right) \|\Psi\|_+,
				\end{equation*}
		where we have defined 
		\begin{equation}\label{eq:explicitbound}
		    \tilde{L} = L c^3e^{c^2M} = c^{11}e^{2c^2M}.
		\end{equation}
		That is,
		\begin{equation*}
			\|\Psi_k(t) - \Psi_\ell(t)\|_{-,\ell,t} \leq \tilde{L}\|F_{k,\ell}\|_{L^1(s,t)}\|\Psi\|_+.\qedhere
		\end{equation*}
	\end{proof}
	{
	We end this subsection by introducing the following result which shows that convergence in the norm of $\mathcal{B}(\H^+,\H^-)$ implies strong convergence in $\H$.
		
	\begin{proposition}\label{prop:stability_H}
		Let $\{H_n(t)\}_{n\in \mathbb{N}}$, $ t\in I$, 
		be a family of time-dependent Hamiltonians with constant form domain $\mathcal{H}^+$. Suppose that for each $n\in\mathbb{N}$ there exists a piecewise weak solution $U_n(t,s)$  of the Schrödinger equation generated by $H_n(t)$ and that $\lim_{n\to\infty}\norm{U_n(t,s) - U_0(t,s)}_{+,-}=0$. Then $\{U_n(t, s)\}$ converges strongly to $U_0(t, s)$ in $\mathcal{H}$.
	\end{proposition}
	}	
	\begin{proof}
		Suppose that $\Psi,\Phi\in\H^+$. Then 
		$$\left|\scalar{\Psi}{\left(U_n(t,s) - U_0(t,s)\right)\Phi}\right|\leq \norm{\Psi}_+\norm{\Phi}_+\norm{U_n(t,s) - U_0(t,s)}_{+,-}.$$
		By the density of $\H^+$ in $\H$ and the unitarity of $U_n(t,s)$ for all $n\in \mathbb{N}$, this implies that $U_n(t,s)\Phi$ converges weakly to $U_0(t,s)\Phi$. Since $\norm{U_n(t,s)\Phi}=\norm{U_0(t,s)\Phi} =\norm{\Phi}$, the latter implies the strong convergence of $U_n(t,s)$ to $U_0(t,s)$.
	\end{proof}
	
This proposition will be used in combination with the stability theorems to obtain the approximate controllability results of Section~\ref{sec:applications}.

	\subsection{Form-linear Hamiltonians}
	For many applications, and in particular for the control of bilinear quantum systems, the time-dependent structure of the Hamiltonians appears through time-dependent factors modulating the intensity of some fixed operators. In this section we are going to focus on this particular case.	
	
	\begin{definition} \label{def:formlinear_hamiltonian}
		Let $I\subset{\R}$ be a compact interval, $H(t)$, ${t \in I}$, a time-dependent Hamiltonian with constant form domain $\mathcal{H}^+$, cf.\ Definition~\ref{def:hamiltonian_const_form}, and let $\mathcal{F}$ be a subset of the real-valued, measurable functions on $I$.
		We say that $H(t)$ is a \emph{form-linear, time-dependent Hamiltonian with space of coefficients} $\mathcal{F}$ if there is a finite collection of Hermitian sesquilinear forms densely defined on $\mathcal{H}^+$, $h_0, h_1, \dots, h_N$, and functions $f_1, f_2, \dots, f_N \in \mathcal{F}$ such that the sesquilinear form $h_t$ associated with $H(t)$ is given by
		\begin{equation*}
			h_t = h_0 + \sum_{i=1}^N f_i(t) h_i.
		\end{equation*}
		The collection $\{ h_i\}_{i=0}^N$ is called the \emph{structure} of $H(t)$, while $\{f_i\}_{i=1}^N$ are called the \emph{coefficients} of $H(t)$.
	\end{definition}
	
	As an immediate consequence of Theorem~\ref{thm:kisynski}, we have the following result on the existence of dynamics for form-linear, time-dependent Hamiltonians.
	\begin{proposition} \label{prop:dynamics_formlinear}
		For every form-linear, time-dependent Hamiltonian whose space of coefficients $\mathcal{F}$ is a subset of $C_{\mathrm{pw}}^1(I)$ there exists a strongly continuous unitary propagator that satisfies~\ref{cor:kisynski_i},~\ref{cor:kisynski_ii} of Corollary~\ref{corol:dynamics_piecewise}.
		If $\mathcal{F} \subset C^2(I)$, then the unitary propagator satisfies~\ref{thm:kisynski_i}-\ref{thm:kisynski_iv} of Theorem~\ref{thm:kisynski}.
	\end{proposition}
	\begin{proof}
		A direct application of Corollary~\ref{corol:dynamics_piecewise} and Theorem~\ref{thm:kisynski} yields the results. 
	\end{proof}
	
	\begin{proposition} \label{prop:formlinear_equiv_norms} 
		For every $n \in \mathbf{N}\subset \mathbb{N}$, let $H_n(t)$, $t\in I$, be a form-linear time-dependent Hamiltonian with structure $\{ h_i\}_{i=0}^N$ and coefficients $\{f_{n,i}\}_{i=1}^N$ such that $\sup_{n,t} |f_{n,i}(t)| < \infty$ for every $i$.
		Assume also that the Hamiltonians $\{H_n(t)\}_{n\in \mathbf{N}}$ have the same lower bound $m$. Then there exists $c$ independent of $n$ and $t$ such that, for every $n \in \mathbf{N}$ and $t \in I$, the associated scales of Hilbert spaces satisfy
		\begin{equation*}
			c^{-1} \|\argdot\|_{\pm,n,t} \leq \|\argdot\|_\pm \leq c \|\argdot\|_{\pm,n,t},
		\end{equation*}
		where we have defined $\|\argdot\|_{\pm}:=\|\argdot\|_{\pm,n_0,t_0}$, for some fixed time $t_0\in I$ and $n_0 \in \boldsymbol{N}$.
	\end{proposition}
	\begin{proof}
		Let $A_n(t) := H_n(t) + m + 1$ for $n\in\mathbf{N}$.
		We have 
		\begin{equation*}
			\langle \Psi, \Phi \rangle_{+,n,t} = \langle A_n(t)^{\sfrac{1}{2}} \Psi, A_n(t)^{\sfrac{1}{2}} \Phi \rangle,
		\end{equation*}
		cf.\ Definition~\ref{def:scales_hilbert}.
		By the Closed Graph Theorem, the operator defined by $T_n(t) \coloneqq A_n(t)^{\sfrac{1}{2}} A_0(t_0)^{-\sfrac{1}{2}}$ is a bounded operator on $\mathcal{H}$ and, for $\Phi\in\mathcal{H}^+$,
		\begin{equation}\label{eq:bound_for_equivalence}
			\|\Phi\|_{+,n,t} = \|T_n(t) A_0(t_0)^{+\sfrac{1}{2}} \Phi\| \leq \|T_n(t)\| \|\Phi\|_{+}.
		\end{equation}
		Analogously one gets  $\|\Phi\|_{+} \leq \|T_n(t)^{-1}\| \|\Phi\|_{+,n,t}$ for $\Phi \in \mathcal{H}^+$.
		
		Using $T_n(t)^*$, the adjoint of the operator $T_n(t)$, one can prove similar inequalities:
		\begin{equation} \label{eq:Tdagger_minus_norm}
			\|\Phi\|_{-,n,t} \leq \|(T_n(t)^*)^{-1}\| \|\Phi\|_-, \quad
			\|\Phi\|_- \leq \|T_n(t)^*\| \|\Phi\|_{-,n,t}.
		\end{equation}
		Taking $c_{n,t} = \max \{\|T_n(t)\|^{-1}, \|T_n(t)^{-1}\|\}$ it follows
		\begin{equation} \label{eq:norm_equiv_item1}
			c_{n,t}^{-1} \|\Phi\|_{+,n,t} \leq \|\Phi\|_+ \leq c_{n,t} \|\Phi\|_{+,n,t},
			\qquad \forall \Phi \in \mathcal{H}^+.
		\end{equation}
		
		We will show now that the previous inequalities also hold with constants independent of $n$ and $t$.
		By Equation~\eqref{eq:norm_equiv_item1} we have $c_{n,t}^{-2} \|\Phi\|_+^{2} \leq |\langle \Phi, \Phi \rangle_{+,n,t}| \leq c_{n,t}^2 \|\Phi\|_+^2$; therefore, $\langle \Phi, \Psi \rangle_{+,n,t}$ is a closed, positive, bounded sesquilinear form on the Hilbert space $(\mathcal{H}^+, \norm{\argdot}_+)$. By Riesz's Theorem one can define a positive, bounded, self-adjoint operator such that
		\begin{equation*}
			\langle \Phi, \Psi \rangle_{+,n,t} = \langle Q_n(t)\Phi, \Psi \rangle_+.
		\end{equation*}
		Moreover $Q_n(t)$ is bounded away from the origin by $c_{n,t}^{-2}>0$ and therefore boundedly invertible.
		
		Then
		\begin{align*}
			\|\Phi\|_{+,n,t}^2 &= h_0(\Phi,\Phi) + (m+1) \|\Phi\|^2 + \sum_{i=1}^N f_{n,i}(t) h_i(\Phi,\Phi) \\
			&\leq |h_0(\Phi,\Phi)| + (m+1) \|\Phi\|^2 + \max_i \sup_{n,t} |f_{n,i}(t)| \sum_{i=1}^N |h_i(\Phi, \Phi)|.
		\end{align*}
		Notice that $ \max_i \sup_{n,t} |f_{n,i}(t)|$ is finite by assumption and therefore
		$\sup_{n,t} \langle Q_n(t)\Phi, \Phi \rangle_+ < \infty$. Then, by the Uniform Boundedness Principle there is $K_1$ such that $\|Q_n(t)^{\sfrac{1}{2}}\|_{+} \leq K_1$.
		This implies
		\begin{equation*}
			\|\Phi\|_{+,n,t} \leq K_1 \|\Phi\|_+.
		\end{equation*}
		By hypothesis we have $A_n(t)>1$, and therefore 
		\begin{equation*}
			\|\Phi\|_{+,n,t}^2 \geq \|\Phi\|^2 = \frac{\|\Phi\|^2}{\|\Phi\|_+^2} \|\Phi\|_+^2.
		\end{equation*}
		Hence, for every $\Phi\in\H^+$ there is a constant $K_{2,\Phi}$, independent of $n$ and $t$, such that $\norm{Q_n(t)^{\sfrac{1}{2}}\Phi}_+\geq K_{2,\Phi}\norm{\Phi}_+$ and hence 
		$$\norm{Q_n(t)^{-\sfrac{1}{2}}\Phi}_+\leq K_{2,\Phi}^{-1}\norm{\Phi}_+$$
		
		By the Uniform Boundedness Principle we have $\sup_{n,t} \|Q_n(t)^{-\sfrac{1}{2}}\|_+ = K_2 < \infty$, and
		\begin{align*}
			\|\Phi\|_+^2 &= \langle Q_n(t)\Phi, Q_n(t)^{-1}\Phi \rangle_+ = \langle \Phi, Q_n(t)^{-1}\Phi \rangle_{+,n,t} \\
			&\leq \|\Phi\|_{+,n,t} \|Q_n(t)^{-1} \Phi\|_{+,n,t}
			= \|\Phi\|_{+,n,t} \langle \Phi, Q_n(t)^{-1}\Phi \rangle_{+}^{\sfrac{1}{2}}
			\leq K_2 \|\Phi\|_{+,n,t}\|\Phi\|_+.
		\end{align*}
		This shows that $\|\Phi\|_+ \leq K_2 \|\Phi\|_{+,n,t},$
		and taking $c = \max\{K_1,K_2\}$ it follows
		\begin{equation*}
			c^{-1} \|\Phi\|_{+,n,t} \leq \|\Phi\|_+ \leq c \|\Phi\|_{+,n,t}.
		\end{equation*}
		Finally, we have
		\begin{equation*}
			\|T_n(t)\Phi\| = \|A_0(t_0)^{-\sfrac{1}{2}} \Phi\|_{+,n,t} \leq c \|A_0(t_0)^{-\sfrac{1}{2}}\Phi\|_+ = c \|\Phi\|,
		\end{equation*}
		and therefore $\|T_n(t)\| \leq c$ and similarly $\|T_n(t)^{-1}\| \leq c$.
		By Eq.~\eqref{eq:bound_for_equivalence} and Eq.~\eqref{eq:Tdagger_minus_norm} the statement follows.    
	\end{proof}
	
	\begin{proposition} \label{prop:formlinear_A4}
		For every $n \in \mathbf{N}\subset \mathbb{N}$, let $H_n(t)$, $t\in I$, be a form-linear time-dependent Hamiltonian with structure $\{ h_i\}_{i=0}^N$ and coefficients $\{f_{n,i}\}_{i=1}^N \subset C_{\mathrm{pw}}^1(I)$.
		Let $\{I_j\}_{j=1}^\nu$ be a collection of open disjoint intervals satisfying $I = \bigcup_j \overline{I_j}$ and such that $f_{n,i}$, $1 \leq i \leq N$, is differentiable on every $I_j$, $1 \leq j \leq \nu$.
		Assume that the derivatives of the coefficients satisfy $\sup_n \sum_{j=1}^\nu \|f_{n,i}'\|_{L^1(I_j)} < \infty$ and that there exists a constant $K$ such that $h_i(\Psi, \Phi) \leq K \|\Psi\|_+ \|\Phi\|_+$ for every $i=1,\dots,N$.
		Then the functions
		\begin{equation*}
			\tilde{C}_n(t) \coloneqq \sup_{\substack{\Psi, \Phi \in \mathcal{H}^+ \\ \|\Psi\|_+ = \|\Phi\|_+ = 1}} \left|\frac{\d}{\d t} h_{n,t}(\Psi, \Phi)\right|
		\end{equation*}
		satisfy 
		\ref{assump:L1_deriv_op_bound}.
	\end{proposition}
	\begin{proof}
		We only need to show that there is $L > 0$ such that $\sum_j \|\tilde{C}_n\|_{L^1(I_j)} \leq L$ for every $n$ (see Remark~\ref{remark:stability_assump}).
		For $\Psi,\Phi \in \mathcal{H}^+$, one has
		\begin{equation*}
			\left| \frac{\d}{\d t} h_{n,t}(\Psi, \Phi) \right| \leq \sum_{i=1}^N |f_{n,i}'(t)| |h_i(\Psi, \Phi)|
			\leq K \|\Psi\|_+ \|\Phi\|_+ \sum_{i=1}^N |f_{n,i}'(t)|.
		\end{equation*}		
		Hence,
		\begin{equation*}
			\sum_{j=1}^\nu\|\tilde{C}_n\|_{L^1(I_j)} \leq K \sum_{i=1}^N \sum_{j=1}^\nu \|f_{n,i}'\|_{L^1(I_j)} \leq  KN \max_i \sup_n\sum_{j=1}^\nu \|f_{n,i}'\|_{L^1(I_j)} =: L.\qedhere
		\end{equation*}
	\end{proof}
	
	\begin{theorem} \label{thm:stability_formlinear}
		Let $I \subset \mathbb{R}$ be a compact interval.
		For every $n \in \mathbf{N} \subset \mathbb{N}$, let $H_n(t)$, $t\in I$, be a form-linear time-dependent Hamiltonian with structure $\{h_i\}_{i=0}^N$ and coefficients $\{f_{n,i}\}_{i=1}^N \subset C_{\mathrm{pw}}^2(I)$.
		For any $n \in \boldsymbol{N}$, denote by $\{I_{n,j}\}_{j=1}^{\nu_n}$ the family of open intervals such that $f_{n,i}$, $1 \leq i \leq N$, is differentiable on every $I_{n,j}$, $1 \leq j \leq m_n$.
		Assume that:
		\begin{enumerate}[label=\textit{(\roman*)},nosep,leftmargin=*]
			\item The Hamiltonians $\{H_n(t)\}_{n\in \mathbf{N}}$ have the same lower bound $m$.
			\item For every $i = 1, \dots, N$, it holds $\sup_{n,t} |f_{n,i}(t)| < \infty$.
			\item For every $i = 1, \dots, N$, it holds $\sup_n \sum_j \|f_{n,i}'\|_{L^1(I_{n,j})} < \infty$.
			\item There is $K$ such that $h_i(\Psi, \Phi) \leq K \|\Psi\|_+ \|\Phi\|_+$ for every $i=0,\dots,N$.
		\end{enumerate}
		Then the family of form-linear time-dependent Hamiltonians $\{H_n(t)\}_{n \in \boldsymbol{N}}$ satisfies Assumption~\ref{assump:stability_assumptions} and for $j,k\in\mathbf{N}$, and $s,t \in I$, there exists a constant $L$ independent of $j$, $k$, $t$ and $s$ such that
		\begin{equation*}
			\|U_j(t,s) - U_k(t,s)\|_{+,-} \leq L \sum_{i=1}^N \|f_{j,i} - f_{k,i}\|_{L^1(s,t)}.
		\end{equation*}		
		If, in addition, for every $i$ we have $f_{n,i} \to f_{0,i}$ in $L^1(I)$ then, for every $s,t \in I$, $U_n(t,s)$ converges strongly to $U_0(t,s)$ uniformly on $t,s$ and also in the sense of $\mathcal{B}(\mathcal{H}^+, \mathcal{H^-})$.
	\end{theorem}
	\begin{proof}
		By hypothesis,~\ref{assump:uniformbound} and~\ref{assump:differentiability_forms} hold.
		The condition~\ref{assump:norm_equiv} follows from Proposition~\ref{prop:formlinear_equiv_norms}, and~\ref{assump:L1_deriv_op_bound} from Proposition~\ref{prop:formlinear_A4}.
		Applying Theorem~\ref{thm:stability_bound} and Proposition~\ref{prop:stability_H}, the statement follows.
	\end{proof}	
	
	%
	\section{Applications to Quantum Control}\label{sec:applications}	
	
	In this section we shall introduce the problem of control of quantum mechanical systems, with main emphasis in the situation in which the space of (pure) states of the system, i.e.\ the space of rays of a complex Hilbert space, is of infinite dimension. We shall consider only complex separable Hilbert spaces $\H$. A typical quantum control problem is the one defined by the following non-autonomous Schrödinger equation: 
	\begin{equation}\label{eq:bilinear}
	i\frac{\mathrm{d}}{\mathrm{d}t}\Psi(t) = H_0\Psi(t) + \sum_{i=1}^pu_i(t)H_i\Psi(t),
	\end{equation}
	where $\{H_0,H_1,\dots,H_p\}$ is a family of closed, densely defined symmetric operators on $\H$ and $\mathbf{u} = (u_1,\dots,u_p)\colon[0,T]\to\mathcal{U}\subset\R^p$ is a family of control functions that can be used to drive the state of the system. In this setting, the operator $H_0$ governs the evolution of the system in the absence of interactions and is usually referred as the drift. The operators $H_i$, $i = 1,\dots,p$ represent different interactions available to the experimenters whose intensity at a given time is determined by the value of the functions $\mathbf{u}$. This control problem is usually referred to as a bilinear control problem. The well-posedness  of the aforementioned problem is by no means guaranteed, as it is required that the operator at the right hand side of Eq.~\eqref{eq:bilinear} is a self-adjoint operator and this typically imposes compatibility conditions between the domains of the different operators and bounds on the family of controls $\mathcal{U}$.
	We shall take profit of the results about form-linear Hamiltonians introduced in Section~\ref{sec:form_dynamics}, which will allow us to handle the well-posedness of the problem as well as to provide several controllability results as an application of the stability results obtained.
	
	In this section we shall focus on the following family of control problems, directly related with the form-linear Hamiltonians introduced in the previous section. 
	
	{
	\begin{definition}
                 Let $\H^+$ be a Hilbert space densely and continuously embedded in $\H$,  $H_0$ a positive self-adjoint operator in $\H$ with compact resolvent and form domain $\H^+$. Consider the scale of Hilbert spaces $\H^+\subset \H \subset \H^-$, cf. Definition~\ref{def:scales_hilbert} and Remark~\ref{rem:generic_scale}. Let $\mathcal{U}\subset\R^p$ be a bounded set and for $i=1,\dots,p$,  $H_i:\H^+\to\H^-$ be bounded and symmetric operators. Assume that there exists $m\geq0$ such that for all $\mathbf{u}=(u_1,\dots,u_p) \in \mathcal{U}$ 
$$\Bigl\langle{\Phi}{,\,H_0\Phi + \sum_{i=1}^pu_iH_i\Phi}\Bigr\rangle \geq -m\norm{\Phi}^2\,,\quad \forall \Phi\in\H^+.$$                 
Then, the control system
		\begin{equation}\label{eq:form-linear-control}
			i\frac{\mathrm{d}}{\mathrm{d}t}\Psi(t) = H_0\Psi(t) + \sum_{i=1}^pu_i(t)H_i\Psi(t),\quad \mathbf{u}(t)\in \mathcal{U},\; \Psi(t)\in \H^+
		\end{equation}
		is a \emph{form-linear control system}.
	\end{definition}
	}
	
	{Notice that Eq.~\eqref{eq:form-linear-control} is an equation in $\H^-$. By Proposition~\ref{prop:formlinear_equiv_norms} the sesquilinear form associated with $H_0 + \sum_{i=1}^pu_i(t)H_i$, i.e. $h_0 + \sum_{i=1}^pu_i(t)h_i$, with $h_i(\Phi,\Psi) = \scalar{\Phi}{H_i\Psi}$, $i=0,\dots,p$, is closed in $\H^+$ and then 
 a solution $\Psi(t)\in\H^-$, $t\in I$, of the equation above corresponds to a weak solution of the Schrödinger equation, cf. Definition~\ref{def:weak_Schrodinger}. }Indeed, by considering the Hermitian sesquilinear forms and
	by Theorem~\ref{thm:repKato}, there exists a self-adjoint family of operators $H(\mathbf{u}(t))$ on $\H$, with dense domains $\D\bigl(H(\mathbf{u}(t))\bigr)\subset\H^+$, $t\in I$, such that 
	\begin{equation}
		h_0(\Phi,\Psi) + \sum_{i=1}^pu_i(t)h_i(\Phi,\Psi) = \scalar{\Phi}{H(\mathbf{u}(t))\Psi},\quad \Phi\in\H^+,\; \Psi\in\D\bigl(H(\mathbf{u}(t))\bigr),\; t\in I.
	\end{equation}
	Under the conditions of Theorem~\ref{thm:kisynski}, a solution $\Psi(t)$ of Eq.~\eqref{eq:form-linear-control} can also be a strong solution of the Schrödinger equation, in which case it satisfies
	
	\begin{equation}\label{eq:form-linear-control-strong}
		i\frac{\mathrm{d}}{\mathrm{d}t}\Psi(t) = H(\mathbf{u}(t))\Psi(t),\quad \mathbf{u}(t)\in \mathcal{U},\;\Psi(t)\in\D\bigl(H(\mathbf{u}(t))\bigr),\; t\in I.
	\end{equation}
	
	\begin{remark}\label{rem:form_linear_includes_bilinear}
		{Form-linear control systems include, as a particular case, Hamiltonians of the type $H(\mathbf{u}(t)) = H_0 + \sum_{i=1}^nu_i(t)H_i$ 
			with self-adjointness domain $\mathcal{D}\big(H(\mathbf{u}(t))\bigr)={\mathcal{D}_0}$ independent of time: that is, they include in a natural way \textit{bilinear control systems} as a special case.}
		However, they also admit the situation in which the domains of the Hamiltonians $\mathcal{D}\big(H(\mathbf{u}(t))\bigr)$ have an explicit time-dependence, a circumstance that appears in practical situations, e.g. when describing interactions with time-dependent magnetic fields. 
	\end{remark}
	
	The first step in the study of the control of a quantum system is to study its \textit{reachable sets}. That is, given an evolution equation, a space of controls and an initial state, what is the space of possible states reachable through orbits of the evolution equation? Typically, if any state is accessible from any other, one says that the system is controllable. {In the particular context of control of quantum systems}, 
	{one often addresses pure state controllability, that is, whether any state $\Phi\in\mathcal{H}$ can be driven to any other $\Psi\in\H$ by choosing it as initial condition of the Schrödinger equation and properly selecting the controls.}
	Notice that, since the evolution is unitary, this can only happen if $\norm{\Phi}=\norm{\Psi}$. We shall focus on this type of controllability. 
	Other notions of approximate controllability, like mixed state\footnote{Also called density operators or, in finite dimensional systems, density matrices.} approximate controllability or simultaneous approximate controllability, cf. \cite{DAlessandro2007}, can also be addressed by the results presented here:
	the generalisation to these cases can be done in a straightforward way and we shall not introduce these results.
	
	In the following definitions we assume the well-posedness of the referred dynamical systems, {that is, the existence of a unitary propagator $U(t,s)$ solving Eq.~\eqref{eq:form-linear-control}.}
	
	\begin{definition}\label{def:exact_control}
		A quantum control system
		$$i\frac{\mathrm{d}}{\mathrm{d}t}\Psi(t) = H(\mathbf{u}(t))\Psi(t),\quad \mathbf{u}(t)\in\mathcal{U}$$
		is \emph{(pure state) exactly controllable} if, for any $\Phi,\Psi\in\H$ with $\norm{\Phi}=\norm{\Psi}$, there exist $T>0$ and $\mathbf{u}:[0,T]\to\mathcal{U}$ such that the solution $U(t,s)$ of the Schrödinger equation satisfies
		$${U(T,0)\Phi = \Psi}.$$
	\end{definition}
	{While the notion of exact controllability is appropriate in the context of finite-dimensional control, several results show that it is a very restrictive notion in the infinite-dimensional setting~\cite{BallMarsdenSlemrod1982, Turinici2000,  boussaid2020regular}.}
	An appropriate notion of controllability in infinite-dimensional quantum systems is approximate controllability:	
	\begin{definition}\label{def:approx_control}
		A quantum control system
		$$i\frac{\mathrm{d}}{\mathrm{d}t}\Psi(t) = H(\mathbf{u}(t))\Psi(t),\quad \mathbf{u}(t)\in\mathcal{U}$$
		is \emph{(pure state) approximately controllable} if, for any $\epsilon>0$ and any $\Phi,\Psi\in\H$ with $\norm{\Phi}=\norm{\Psi}$, there exists $T>0$ and $\mathbf{u}:[0,T]\to\mathcal{U}$ such that the solution $U(t,s)$ of the Schrödinger equation satisfies
		$$\norm{U(T,0)\Phi - \Psi} < \epsilon.$$
	\end{definition}
	{In many interesting situations, e.g~\cite{BoscainPozzoliSigalotti2021}, controllability results are obtained by means of piecewise constant control functions. However, such controls are not always suitable for practical situations~\cite{ErvedozaPuel2009}. Furthermore, the corresponding Schr\"odinger equation is not well-posed in the case in which the self-adjointness domain depends on time. This happens, for instance, for systems with time-dependent magnetic fields or time-dependent boundary conditions~\cite{BalmasedaPerezPardo2019, BalmasedaLonigroPerezPardo2021, BalmasedaLonigroPerezPardo2022}.}
	
	The next result is a direct application of Theorem~\ref{thm:stability_formlinear} that allows for an extension of the family of control functions from piecewise constant to smooth controls: by the Remark~\ref{rem:form_linear_includes_bilinear}, it can be naturally applied to the Lie--Galerkin control results obtained in~\cite{ChambrionMasonSigalottiEtAl2009, BoscainCaponigroChambrionEtAl2012, caponigro2018exact, boussaid2023controllability} for bilinear control systems.
	
	As in the rest of this article we are going to assume that the time-dependent Hamiltonians are uniformly bounded from below and we shall not repeat this statement in the following results.	
	
	\begin{theorem}\label{thm:smooth_controls}
		A form-linear control system defined by $\{H_0, H_1, \dots, H_p\}$ and $\mathcal{U}\subset\mathbb{R}^p$ is approximately controllable with controls $\mathbf{u}\in C^{\infty}([0,T], \mathcal{U})$, if it is approximately controllable with piecewise constant controls. 
	\end{theorem}
	
	\begin{proof}
		For any $\Phi,\Psi\in\H$ and $\epsilon>0$ we want to find $\mathbf{u}\in C^{\infty}([0,T], \mathcal{U})$, with $T$ possibly depending on $\epsilon$, $\Phi$, $\Psi$, such that the weak solution of the Schrödinger equation generated by $\mathbf{u}$, $U_\mathbf{u}(t,s)$, satisfies $\norm{U_\mathbf{u}(T,0)\Phi - \Psi}<\epsilon$.
		
		By assumption, for any $\Phi,\Psi\in\H$ and $\epsilon>0$ there exists $T$ and $\tilde{u}_i\colon[0,T]\to\mathbb{R}$, $i = 1, \dots, p$ piecewise constant such that the piecewise weak solution $U_{\tilde{\mathbf{u}}}(t,0)\Phi$ of the Schrödinger equation determined by the form-linear Hamiltonian $H_0 + \sum_{i=1}^p\tilde{u}_i(t)H_i$ satisfies $\norm{U_{\tilde{\mathbf{u}}}(T,0)\Phi - \Psi}<\epsilon$.
		
		For any $\mathbf{u}\in C^\infty([0,T],\mathcal{U})$, with $U_\mathbf{u}(t,s)$ being the strong solution of the Schrödinger equation determined by the control ${u}(t)$, one has
		$$\norm{U_\mathbf{u}(T,0)\Phi - \Psi} \leq \norm{\Bigl(U_\mathbf{u}(T,0) - U_{\tilde{\mathbf{u}}}(T,0)\Bigr)\Phi} + \norm{U_{\tilde{\mathbf{u}}}(T,0)\Phi - \Psi}.$$
		Therefore, by Proposition~\ref{prop:stability_H} it is enough to find a sequence of smooth functions $\{\mathbf{u}_k\}\subset C^{\infty}([0,T],\mathcal{U})$ such that 
		$$\lim_{k\to\infty}\norm{U_{\tilde{\mathbf{u}}}(T,0) - U_{\mathbf{u}_k}(T,0)}_{+,-}= 0.$$
		
		For any $k\in\mathbb{N}$ and $i = 1,\dots, p$ there exist $u_{i,k}\in C^{\infty}([0,T])$ such that $\norm{u_{i,k} - \tilde{u}_i}_{L^1([0,T])} < \frac{1}{k}$ and $\operatorname{sup}_{k\in\mathbb{N}}\|{u'_{i,k}}\|_{L^1(0,T)}<\infty$. Now notice that, by choosing the family $\{I_{k,j}\}_{j=1}^{q}$ as the family of open intervals in which all $\tilde{u}_i$ are continuous, the conditions of Theorem~\ref{thm:stability_formlinear} hold and therefore, setting $\mathbf{u}_k = (u_{1,k}, \dots,  u_{p,k})$,
		$$\norm{U_{\tilde{\mathbf{u}}}(t,s) - U_{\mathbf{u}_k}(t,s)}_{+,-} \leq L \sum_{i=1}^p\norm{\tilde{u}_i - u_{i,k}}_{L^1([s,t])},$$
		for any $t,s\in[0,T]$, as we wanted to show.
	\end{proof}
	
	\begin{remark}\label{rem:smooth_implies_strong}
		Notice that Theorem~\ref{thm:smooth_controls} can be applied to extend the controllability results from piecewise weak solutions to strong solutions, cf. Theorem~\ref{thm:kisynski}, even in the cases in which the operator domain depends on time.
	\end{remark}
	
	Next we present an abstract result providing sufficient conditions, that are commonly met in practical scenarios, for 
	{extending finite-dimensional controllability results to the infinite-dimensional Schrödinger equation.}
	
	Recall that compact operators can be approximated by finite-rank operators. Given $T\colon\H^+\to\H^-$ compact, and a sequence $\{P_n\}_{n=1}^\infty$ of finite-rank projections such that for any $\Phi\in\H^{\pm}$ one has, respectively, $\lim_{n\to\infty}\norm{P_n\Phi-\Phi}_{\pm}=0$, then $\lim_{n\to\infty}\norm{T-P_nTP_n} = 0$. We will apply this property, together with the stability result, to obtain a controllability result for form-linear control systems, {proving that such systems are approximately controllable} provided that their finite-dimensional approximations are exactly controllable with controls bounded in the $L^1$-norm.
	
	Since $H_0$ is a positive self-adjoint operator with compact resolvent then $H_0^{1/2}$ has compact resolvent as well, and there exists an orthonormal basis of eigenvectors of $H_0$. Let $P_n$ be the finite-rank orthogonal projector onto the eigenspace spanned by the first $n$ eigenfunctions of $H_0$. Then, for any $n\in\mathbb{N}$, any form-linear control system gives rise to a finite-dimensional quantum bilinear control system:
        \begin{equation}i\frac{\mathrm{d}}{\mathrm{d}t}\Psi(t) = \left( P_nH_0P_n + \sum_{i=1}^p u_i(t)P_nH_iP_n \right)\Psi(t).\tag{$\Sigma$}\end{equation}\label{eq:system_Sigma}
         The system defined by the Schrödinger equation $\Sigma$
	 is an ODE. In this case, since $\operatorname{ran} P_n = n < \infty$, weak and strong solutions of the Schrödinger equation 
	 are equivalent.
	
	{
	\begin{theorem}\label{thm:Lie-Galerkin control}
		Let $\{H_0,H_1,\dots H_p\}$ and $\mathcal{U}$ define a form-linear control system. Assume that the operators $H_i\colon \H^+\to\H^-$, $i = 1,\dots,p$ are compact. 
                 Suppose that there exists $\ell>0$ such that, 
		for each $n$, the finite-dimensional quantum bilinear control system $\Sigma$ is exactly controllable with {piecewise constant} control $\mathbf{u}_n\colon [0,T_n]\to\mathcal{U}$ such that $\norm{u_{i,n}}_{L^1(0,T_n)}<\ell$, $i = 1,\dots,p$. Then, the form-linear control system is approximately controllable {with piecewise constant controls}.
	\end{theorem}
	}
    {Instead of proving Theorem~\ref{thm:Lie-Galerkin control} we are going to prove the next stronger result which implies it.

\begin{theorem}\label{thm:Lie-Galerkin control strong}
		Let $\{H_0,H_1,\dots H_p\}$ and $\mathcal{U}$ define a form-linear control system. Assume that the operators $H_i\colon \H^+\to\H^-$, $i = 1,\dots,p$ are compact. Suppose that for every $n'$ and every $\Phi, \Psi\in P_{n'}\H$ there exists $\ell=\ell(n',\Phi,\Psi)>0$ such that for each $n>n'$, there is $T_n>0$ and a piecewise constant $\mathbf{u}_{n}:[0,T_n]\to\mathcal{U}$ with $\norm{u_{i,n}}_{L^1(0,T_n)}<\ell$, $i = 1,\dots,p$, and such that the weak solution $U_n(t,s)$ of the Schrödinger equation $\Sigma$ satisfies $U_n(T_n,0)\Phi = \Psi$. Then, the form-linear control system is approximately controllable with piecewise constant controls.
	\end{theorem}
    }

	\begin{proof}
		First notice that one can choose 
		$$\norm{\Phi}_+= \sqrt{h_0(\Phi,\Phi) + \norm{\Phi}^2} = \sqrt{\scalar{H_0^{\nfrac{1}{2}}\Phi}{H_0^{\nfrac{1}{2}}\Phi} + \norm{\Phi}^2}.$$
		Let $\{\Phi_n\}_{n\in\mathbb{N}}$ be the orthonormal basis of $H_0$. Then $H_0\Phi_n=\lambda_n\Phi_n$ and 
		$\norm{\Phi_n}_{\pm} = (1+\lambda_n)^{\pm \nfrac{1}{2}}.$
		
		Define $\tilde{\Phi}_n:=\frac{\Phi_n}{\norm{\Phi_n}_+} = (1+\lambda_n)^{- \nfrac{1}{2}}\Phi_n$ and $\hat{\Phi}_n:=\frac{\Phi_n}{\norm{\Phi_n}_-} = (1+\lambda_n)^{\nfrac{1}{2}}\Phi_n$.  Then $\{\tilde{\Phi}_n\}_{n\in\mathbb{N}}$ and $\{\hat{\Phi}_n\}_{n\in\mathbb{N}}$ are orthonormal bases in $\H^+$ and $\H^-$ respectively. A straightforward calculation shows that 
		$$P_n\Psi = \sum_{j=1}^n \scalar{\Phi_j}{\Psi}\Phi_j = \sum_{j=1}^n \scalar{\tilde{\Phi}_j}{\Psi}_+\tilde{\Phi}_j = \sum_{j=1}^n \scalar{\hat{\Phi}_j}{\Psi}_-\hat{\Phi}_j,$$
		which proves $\lim_{n\to\infty}\norm{P_n\Psi - \Psi}_{\pm} = 0$ for $\Psi\in\H^{\pm}$ respectively.
		
		Since $\norm{H_0^{\pm\nfrac{1}{2}} e^{itH_0}\Phi} = \norm{e^{itH_0}H_0^{\pm\nfrac{1}{2}} \Phi} = \norm{H_0^{\pm\nfrac{1}{2}}\Phi} $ for $\Phi\in\H^{\pm}$, respectively, we have
		$$\norm{e^{itH_0}}_+ = \sup_{\norm{\Phi}_+=1} \norm{e^{itH_0}\Phi}_+ = \sup_{\norm{\Phi}_+=1} \sqrt{\norm{H_0^{\nfrac{1}{2}} e^{itH_0}\Phi}^2 + \norm{e^{itH_0}\Phi}^2} = \sup_{\norm{\Phi}_+=1} \norm{\Phi}_+ = 1,$$
		and equivalently for $\norm{e^{itH_0}}_-=1$.
		
                 Let $\epsilon>0$, $\Phi,\Psi\in\H$ with $\norm{\Phi}=\norm{\Psi}$. Then there exist $n'\in\mathbb{N}$, $\Psi_{n'},\Phi_{n'}\in P_{n'}\H$ with $\norm{\Phi_{n'}}=\norm{\Psi_{n'}}$ such that $\norm{\Phi_{n'} - \Phi}<\frac{\epsilon}{3}$ and $\norm{\Psi_{n'} - \Psi}<\frac{\epsilon}{3}$. Also, by hypothesis, there exists $\ell>0$ such that for all $n>n'$ there exists $T_n>0$ and $\mathbf{u}_{n}:[0,T_n]\to\mathcal{U}$ piecewise constant with  $\norm{u_{i,n}}_{L^1(0,T_n)}<\ell$, $i = 1,\dots,p$, such that the solution $U_n(t,s)$ of the Schrödinger equation
                $$i\frac{\mathrm{d}}{\mathrm{d}t}\Psi(t) = \left( P_nH_0P_n + \sum_{i=1}^p u_{i,n}(t)P_nH_iP_n       \right)\Psi(t)$$
                satisfies $U_n(T,0)\Phi_{n'} = \Psi_{n'}$. Now, denote by $U(t,s)$ the weak solution of the Schrödinger equation
                $$i\frac{\mathrm{d}}{\mathrm{d}t}\Psi(t) = \left( H_0 + \sum_{i=1}^p u_{i,n}(t)H_i       \right)\Psi(t).$$
                It holds:
                \begin{align}\label{eq:uniform_lie-galerkin_bound}
			\norm{U(T_n,0)\Phi-\Psi} &\leq \norm{U(T_n,0)(\Phi-\Phi_{n'})} + \norm{\Psi-\Psi_{n'}}+  \norm{U(T_n,0)\Phi_{n'} - U_n(T_n,0)\Phi_{n'}} + \norm{U_n(T_n,0)\Phi_{n'} - \Psi_{n'}}\\
			& \leq \frac{2\epsilon}{3} + \norm{U(T_n,0)\Phi_{n'} - U_n(T_n,0)\Phi_{n'}}.\nonumber 
		\end{align}
		Hence, it is enough to show that $U_n$ converges strongly to $U$ and therefore, by Proposition~\ref{prop:stability_H}, it suffices to prove that
		$$\lim_{n\to\infty}\norm{U(T_n,0) - U_n(T_n,0)}_{+,-} = 0.$$


		To prove the result we will need the interaction picture. Define $\tilde{U}(t,s) := e^{itH_0}U(t,s)e^{-isH_0}$ and notice that it is the unitary propagator solving the Schrödinger equation generated by $\tilde{H}(t) = \sum_{i = 1}^p u_{i,n}(t) e^{itH_0}H_ie^{-itH_0}$. In the same way, define $\tilde{U}_n(t,s) := e^{itH_0}U_n(t,s)e^{-isH_0}$. Having into account that $P_nU_n(t,s) = U_n(t,s)$ and that $P_n$ commutes with $H_0$ and $e^{itH_0}$, a straightforward calculation shows that $\tilde{U}_n(t,s)$ is the unitary propagator solving the Schrödinger equation generated by $\tilde{H}_n(t) = \sum_{i = 1}^p u_{i,n}(t) e^{itH_0}P_nH_iP_ne^{-itH_0}$.
		
		Now we have 
		\begin{align*}
			\norm{U_n(t,s) - U(t,s)}_{+,-} &\leq \norm{e^{-itH_0}}_-  \norm{\tilde{U}_n(t,s) - \tilde{U}(t,s)}_{+,-} \norm{e^{isH_0}}_+ \\
			& \leq L \int_{s}^t \sum_{i=1}^p |u_{i,n}(t)| \norm{e^{itH_0}(P_nH_iP_n - H_i)e^{-itH_0}}_{+,-} \mathrm{d}t \\
			&  \leq L \sum_{i=1}^p \norm{P_nH_iP_n - H_i}_{+,-} \int_0^{T_{n}}|u_{i,n}(t)|\mathrm{d}t,
		\end{align*}
		where we have used the stability Theorem~\ref{thm:stability_formlinear}. Notice that the time $T_n$ can depend on $n$, but $ \int_0^{T_{n}}|u_{i,n}(t)|\mathrm{d}t$ is uniformly bounded by assumption. The constant from Theorem~\ref{thm:stability_formlinear} {is explicitly $L = c^{11}e^{2c^2M}$}, cf.\ Eq.\eqref{eq:explicitbound}, where $c$ depends only on the supremum of the controls, which are also uniformly bounded and 
		$M = \sup_{n,i} \sum_{j=1}^p\norm{u'_{i,n}(t)}_{L^1(I_{n,j})}$, where $I_{n,j}$ are the open intervals in which the control functions $\{u_{i,n}\}_{i=1}^p$ are differentiable. Since, for each $n$, $u_{i,n}:[0,T_n]\to\mathbb{R}$ is a piecewise constant function, it follows that $M=0$ independently of $n$ and therefore $L$ is independent of $n$. Hence, the right hand side of Eq.~\eqref{eq:uniform_lie-galerkin_bound} can be chosen arbitrary small as we wanted to show. 
	\end{proof}
	
	\begin{remark}\label{rem:our_stability_improves_kato}
		Notice that, in the proof above, we have used the fact that the controls are piecewise constant and that the stability theorems~\ref{thm:stability_bound} and~\ref{thm:stability_formlinear} apply to this case with a uniform bound on the derivatives of the controls without requiring
	{the derivatives to be either uniformly bounded or having a uniformly bounded $L^1$-norm,}
	{differently from the stability results in~\cite[Theorem III]{Kato1973} and~\cite{Sloan1981}.}
	\end{remark}
	\rev{
	}
	
	{Finally, it will be useful to remark that} the assumptions about the compactness of the resolvent of $H_0$ and the compactness of the operators $H_i\colon\H^+\to\H^-$ are in fact connected. First, recall the equivalence of the following assertions about a self-adjoint operator $H_0$, cf.\ \cite[Proposition 5.12]{schmudgen2012unbounded}:
	\begin{itemize}
		\item $H_0$ has compact resolvent
		\item $H_0^{\nfrac{1}{2}}$ has compact resolvent
		\item The form domain of $H_0$, $\H^+ = \D(|H_0|^{\nfrac{1}{2}})$ is compactly embedded in $\H$.
	\end{itemize}
	The conditions above hold in many practical situations. For example: given a compact Riemannian manifold $\Omega$, e.g. a bounded domain of $\R^n$, that satisfies certain technical assumptions on the boundary, then the Rellich--Kondrachov Theorem,~\cite[Section~6.3]{AdamsFournier2003}, characterises several compact embeddings between Sobolev spaces. Of particular importance are embeddings of the type $\H^{k}(\Omega){\,\hookrightarrow\,} L^2(\Omega)$, $k\geq 0$, where $\H^{k}(\Omega)$ stands for the Sobolev space of order $k$, which are natural domains for differential operators on $\Omega$. In the context of quantum control, $H_0$ can be a self-adjoint realisation of the Laplace--Beltrami or Dirac operator, whose form domain is typically a closed subset of $\H^1(\Omega)$ or $\H^{\nfrac{1}{2}}(\Omega)$ respectively, e.g.~\cite{asorey2013edge, asorey2016edge, Davies1995, IbortLledoPerezPardo2015, PerezPardo2017}. Notice that, even though Dirac operators are not semibounded operators, very often they can be reduced to a direct sum of a positive and a negative operator and one can study controllability in either of the closed subspaces associated to the reduction.
	The following proposition gives sufficient conditions for the operators $H_i\colon\H^+ \to \H^-$ to be compact.
	
	\begin{proposition}\label{prop:compact_emb}
		Let $\H^+\subset \H \subset \H^-$ be a scale of Hilbert spaces and suppose that $\H^+$ is compactly embedded in $\H$. Then any operator $T\in \mathcal{B}(\H^+,\H)$
{is compact as an operator from $\H^+$ to $\H^-$.}
	\end{proposition}
	
	\begin{proof}
		Let $J\colon\H^+\to\H^-$ be the canonical isomorphism that relates the norms in the scale of Hilbert spaces, i.e. $\norm{J^{\pm1} \Phi} = \norm{\Phi}_{\pm}$, for $\Phi \in \H^{\pm}$ respectively. Let $\{\xi_j\}_{j}$ be a bounded sequence in $\H$ and define $\Phi_j:= J^{-1}\xi_j$ which is a bounded sequence in $\H^+$. By the compact embedding there is a subsequence $\{\Phi_{j_n}\}_{n}$ which is convergent in $\H$, hence $\{J^{-1}\xi_{j_n}\}_n$ is a convergent sequence in $\H$ and therefore $J^{-1}$ is a compact operator in $\mathcal{B}(\H)$.
		
		Let $\{\Phi_j\}$ be a bounded sequence in $\H^+$ and, by the compactness of $J^{-1}$, there is a subsequence such that $\{J^{-1}T\Phi_{j_n}\}_{n}$ is convergent in $\H$. Since 
		$\norm{T\Phi}_- = \norm{J^{-1}T\Phi}$,
		it follows that $\{T\Phi_{j_n}\}_{n}$ is a convergent sequence in $\H^-$, as we wanted to show.
	\end{proof}

{From Proposition~\ref{prop:compact_emb} we can give the following straightforward corollary of Theorem~\ref{thm:Lie-Galerkin control}. Notice that, because of the natural inclusions $\mathcal{B}(\H^+)\subset \mathcal{B}(\H^+,\H)$ and $\mathcal{B}(\H)\subset \mathcal{B}(\H^+, \H)$, the cases of bounded operators on $\H$ and bounded operators on $\H^+$ are also included.
}

		{
	\begin{corollary}\label{cor:Lie-Galerkin control}
		Let $\{H_0,H_1,\dots H_p\}$ and $\mathbf{u}\in\mathcal{U}$, define a form-linear control system. Assume that the operators $H_i\colon \H^+\to\H^-$, $i = 1,\dots,p$ are such that $H_i\in \mathcal{B}(\H^+,\H)$. 
                 Suppose that there exists $\ell>0$ such that, 
		for each $n$, the finite-dimensional quantum bilinear control system defined by $\{h_i^{(n)}\}_{i=0}^p$ is exactly controllable with {piecewise constant} control $\mathbf{u}\colon [0,T_n]\to\mathcal{U}$ such that $\norm{u_i}_{L^1(0,T_n)}<\ell$, $i = 1,\dots,p$. Then, the form-linear control system is approximately controllable {with piecewise constant control}.
	\end{corollary}
	}
There is also a stronger version of this Corollary under the weaker assumptions for the controls given in Theorem~\ref{thm:Lie-Galerkin control strong}. We shall not state it explicitly. 
	
	
	\section{Final Remarks}\label{sec:remarks}	
	
	{The conditions on the controls for the finite-dimensional approximations---namely, bounded controls which also have a bounded $L^1$-norm---that appear in Theorem~\ref{thm:Lie-Galerkin control} and Theorem~\ref{thm:Lie-Galerkin control strong}
{
are indeed encountered as natural assumptions for the control of the finite-dimensional approximations of bilinear quantum control systems, as shown in~\cite{BoscainCaponigroChambrionEtAl2012, boussaid2013weakly, caponigro2018exact, boussaid2023controllability}. As such, these results may be applied to numerous systems of practical interest as a tool to transfer controllability from the finite-dimensional level to the infinite-dimensional one. Furthermore, while we only took into account controllability between pure states, straightforward generalisations to other types of controllability, like simultaneous controllability or mixed state controllability, are possible.}	
{Besides, both Theorem~\ref{thm:smooth_controls} and Theorem~\ref{thm:Lie-Galerkin control} naturally allow the Hamiltonians to exhibit time-dependency in their domains, i.e. in the boundary conditions, and therefore can handle the situations described, for instance, in~\cite{BalmasedaPerezPardo2019}, and study controllability including a time-dependence in the parameters of the boundary conditions~\cite{BalmasedaDiCosmoPerezPardo2019, BalmasedaLonigroPerezPardo2021}.}

{Our {results} can already provide straightforward generalisations of existing controllability results. For one,} the combination of Theorem~\ref{thm:Lie-Galerkin control} with Proposition~\ref{prop:compact_emb}, Corollary~\ref{cor:Lie-Galerkin control}, provides a generalisation of the results on~\cite[Theorem 4]{boussaid2013weakly} for the case $k=1$, in which the $k$-weakly coupled assumption, which in our case becomes $\norm{H_i\Phi}_+<\norm{\Phi}_+$, can be dropped. 
	Besides, from the discussion before Proposition~\ref{prop:compact_emb} and Corollary~\ref{cor:Lie-Galerkin control}, it follows that, for quantum systems defined on compact manifolds, any interaction term in $\mathcal{B}(\H^+, \H)$, which includes $\mathcal{B}(\H)$, allows the Lie--Galerkin approximation to be applied. This increases the class of admissible interactions and allows for unbounded interactions. In particular, in the case in which $H_0$ is a second order elliptic operator, first order differential operators are also admissible interactions, differently from the results in~\cite{nersesyan2010global}. 
	Notice that in this case semiboundedness is not guaranteed and it needs to be ensured in order to apply the stability results, for instance, as it was done in \cite{BalmasedaLonigroPerezPardo2022}.
	The solutions of the Schrödinger equation still satisfy $\norm{U(t,s)\Phi}_+ < K\norm{\Phi}_+$, cf. Lemma~\ref{lemma:propagators_bound_simon}, and can be made into strong solutions of the Schrödinger equation by means of Theorem~\ref{thm:smooth_controls}, see also Remark~\ref{rem:smooth_implies_strong} and Theorem~\ref{thm:kisynski}. However, the approximability with respect to the stronger norm $\norm{\cdot}_+$, i.e., $\norm{U(T,0)\Phi-\Psi}_+ < \epsilon$, is not recovered. Theorem~\ref{thm:smooth_controls} establishes that any quantum bilinear control system that is approximately controllable by piecewise constant controls is also approximately controllable by smooth controls, thus complementing the results in \cite{caponigro2018exact} and extending them to the case of form-linear Hamiltonians.
	
	It is important to stress that all controllability results of Section~\ref{sec:applications} crucially rely on the sharper bound found for the stability Theorem~\ref{thm:stability_bound} and could not follow from previous stability results. The result presented in this work applies to the situation in which the generator of the first order evolution equation is a skew-adjoint operator ($iH_0$) on a Hilbert space. A generalisation of these results to the case of quasi-stable families of generators of evolutions on Banach spaces, as treated in~\cite{Kato1973}, is still an open problem.
	
	
	The compactness of the interval $I$, which was required to obtain many of the results of this article, is indeed a common request which is also found in similar results in the literature. We believe that this condition can be relaxed: this would be done at the expense of requiring extra conditions on the asymptotic behaviour of Hamiltonians. A detailed analysis of this conjecture should be carried out.
	
	To conclude, we expect this set of ideas to lead to a more general framework for Lie--Galerkin controllability methods, by allowing the use of more general finite-dimensional approximations than the projections onto the proper subspaces of the drift Hamiltonian $H_0$. This could present advantages for the numerical implementation of optimal control methods for quantum control systems associated with differential operators by allowing, for instance, the use of Finite Element Methods leading to numerical approximations by sparse matrices, see e.g~\cite{IbortPerezPardo2013, LopezYelaPerezPardo2017, PerezPardoBarberoLinanIbort2015} and also in the context of time-dependent boundary conditions.
		
	
	\subsection*{Acknowledgements}
	A.B. and J.M.P.P. acknowledge support provided by the ``Agencia Estatal de Investigación (AEI)'' Research Project PID2020-117477GB-I00, by the QUITEMAD Project P2018/TCS-4342 funded by the Madrid Government (Comunidad de Madrid-Spain) and by the Madrid Government (Comunidad de Madrid-Spain) under the Multiannual Agreement with UC3M in the line of ``Research Funds for Beatriz Galindo Fellowships'' (C\&QIG-BG-CM-UC3M), and in the context of the V PRICIT (Regional Programme of Research and Technological Innovation).
	J.M.P.P acknowledges financial support from the Spanish Ministry of Science and Innovation, through the ``Severo Ochoa Programme for Centres of Excellence in R\&D'' (CEX2019-000904-S).
	A.B.\ acknowledges financial support from the Spanish Ministry of Universities through the UC3M Margarita Salas 2021-2023 program (``Convocatoria de la Universidad Carlos III de Madrid de Ayudas para la recualificación del sistema universitario español para 2021-2023''), and from ``Universidad Carlos III de Madrid'' through Ph.D.\ program grant PIPF UC3M 01-1819, UC3M mobility grant in 2020 and from the EXPRO grant No.\ 20-17749X of the Czech Science Foundation.
	D.L. was partially supported by ``Istituto Nazionale di Fisica Nucleare'' (INFN) through the project ``QUANTUM'' and the Italian National Group of Mathematical Physics (GNFM-INdAM), and acknowledges financial support by European Union – NextGenerationEU (CN00000013 – “National Centre
	for HPC, Big Data and Quantum Computing”).
	He also thanks the Department of Mathematics at ``Universidad Carlos III de Madrid'' for its hospitality.
	
	\printbibliography
	\DeclareFieldFormat{pages}{#1}\sloppy 
	
\end{document}